\documentclass[sn-mathphys,Numbered]{sn-jnl}

\usepackage{graphicx}%
\usepackage{multirow}%
\usepackage{amsmath,amssymb,amsfonts}%
\usepackage{amsthm}%
\usepackage{mathrsfs}%
\usepackage[title]{appendix}%
\usepackage{xcolor}%
\usepackage{textcomp}%
\usepackage{manyfoot}%
\usepackage{booktabs}%
\usepackage{algorithm}%
\usepackage{algorithmicx}%
\usepackage{algpseudocode}%
\usepackage{listings}
\usepackage{mathrsfs}
\usepackage{mathtools}

\usepackage{bbm}
\usepackage{amsthm}
\usepackage{centernot}
\usepackage{caption}
\usepackage{subcaption}
\usepackage{braket}
\usepackage{lastpage}
\usepackage{enumitem}
\usepackage{setspace}
\usepackage{xcolor}
\usepackage[english]{babel} 
\usepackage{dsfont}

\newcommand{\euler}[1]{\text{e}^{#1}}

\newcommand{\supp}{\text{supp}}
\newcommand{\norm}[1]{\left\lVert #1 \right\rVert}
\newcommand{\abs}[1]{\left\lvert #1 \right\rvert}

\newcommand{\ceil}[1]{\left\lceil #1 \right\rceil}

\newcommand{\dom}[1]{\mathcal D\left(#1\right)}

\newcommand{\interior}[1]{%
	{\kern0pt#1}^{\mathrm{o}}%
}
\renewcommand{\braket}[1]{\left\langle#1\right\rangle}
\newcommand*\diff{\mathop{}\!\mathrm{d}}

\newcommand{\R}{\mathbb{R}}

\newcommand{\rr}{\mathcal{R}}

\theoremstyle{thmstyleone}%
\newtheorem{theorem}{Theorem}
\newtheorem{proposition}[theorem]{Proposition}%

\theoremstyle{thmstyletwo}%
\newtheorem{remark}{Remark}%

\theoremstyle{thmstylethree}%

\newtheorem{lemma}[theorem]{Lemma}
\newtheorem{corollary}[theorem]{Corollary}

\begin{document}

\title[Article Title]{Ground State Energy of Dilute Bose Gases in 1D}


\author*[1]{\fnm{Johannes} \sur{Agerskov}}\email{johannes-as@math.ku.dk}

\author[2]{\fnm{Robin} \sur{Reuvers}}\email{robin.reuvers@uniroma3.it}

\author[1]{\fnm{Jan Philip} \sur{Solovej}}\email{solovej@math.ku.dk}

\affil*[1]{\orgdiv{Department of Mathematical Sciences}, \orgname{University of Copenhagen}, \orgaddress{\street{Universitetsparken 5}, \city{Copenhagen \O}, \postcode{DK-2100}, \country{Denmark}}}

\affil[2]{\orgdiv{Dipartimento di Matematica e Fisica}, \orgname{Universit\`a degli Studi Roma Tre}, \orgaddress{\street{ L.go S. L.
			Murialdo 1}, \city{Roma}, \postcode{00146}, \country{Italy}}}


\abstract{We study the ground state energy of a gas of 1D bosons with density $\rho$, interacting through a general, repulsive 2-body potential with scattering length $a$, in the dilute limit $\rho |a|\ll1$. The first terms in the expansion of the thermodynamic energy density are $(\pi^2\rho^3/3)(1+2\rho a)$, where the leading order is the 1D free Fermi gas. This result covers the Tonks--Girardeau limit of the Lieb--Liniger model as a special case, but given the possibility that $a>0$, it also applies to potentials that differ significantly from a delta function. We include extensions to spinless fermions and 1D anyonic symmetries, and discuss an application to confined 3D gases.}




\maketitle

\section{Introduction}
The ground state energy of interacting, dilute Bose gases in 2 and 3 dimensions has long been a topic of study. Usually, a Hamiltonian of the form 
\begin{equation}
	\label{Hgeneral}
	-\sum^N_{i=1}\Delta_{x_i}+\sum_{1\leq i<j\leq N}v(x_i-x_j)
\end{equation}
is considered ($\hbar=2m=1$), in a box $[0,L]^d$ of dimension $d=2,3$, and with a repulsive 2-body interaction $v\geq0$ between the bosons. Diluteness is defined by saying the density $\rho=N/L^d$ of the gas is low compared to the scale set by the scattering length $a$ of the potential (see Appendix C in \cite{lieb2006mathematics} for a discussion, and also Section \ref{SecProofidea} for $d=1$ below). That is, $\rho a^2\ll1$ in 2D, and $\rho a^3\ll1$ in 3D.

In the thermodynamic limit, the diluteness assumption allows for surprisingly general expressions for the ground state energy. Take, for example, the famous energy expansion to second order in $\rho a^3\ll1$ by Lee--Huang--Yang \cite{lee1957eigenvalues}, derived for 3D bosons with a hard core of diameter $a$,
\begin{equation}
	\label{result3D}
	4\pi N\rho^{2/3} (\rho a^3)^{1/3}\left(1+\frac{128}{15\sqrt{\pi}}\sqrt{\rho a^3}+o\left(\sqrt{\rho a^3}\right)\right).
\end{equation}
After early rigorous work by Dyson \cite{dyson1957ground}, Lieb and Yngvason \cite{lieb1998ground} proved that the leading term in this expansion holds for a very general class of potentials $v$, and a similar result was obtained for the second-order term \cite{yau2009second,fournais2020energy,basti2021new,fournais2021energy,basti2024upper, haberberger2023free, haberberger2024upper,fournais2024free}.

The situation is similar in 2D. The leading order in the energy expansion for $\rho a^2\ll1$ derived by Schick \cite{schick1971two} was proved rigorously by Lieb and Yngvason \cite{lieb2001ground}. A second-order term has also been derived and was equally predicted to be general \cite{andersen2002ground,mora2009ground,fournais2019ground}, resulting in the expansion
\begin{equation}
	\label{result2D}
	\frac{4\pi N\rho}{\abs{\ln(\rho a^2)}}\left(1-\frac{\ln{\abs{\ln(\rho a^2)}}}{\abs{\ln(\rho a^2)}}+\frac{C}{\abs{\ln(\rho a^2)}}+o\left(\abs{\ln(\rho a^2)}^{-1}\right)\right),
\end{equation}
for some constant $C$. This was recently shown rigorously \cite{fournais2024ground,fournais2024lower}. 

Remarkably, it seems the existence of a general expansion in 1D was never studied in similar depth. It was, however, suggested in \cite{astrakharchik2010low} by considering two exactly-known special cases, as we will do now as well.

The first is the famous Lieb--Liniger model \cite{lieb1963exact}. Many of its features can be calculated explicitly with Bethe ansatz wave functions, but for our purpose we return to something basic: the ground state energy.
Consider Lieb and Liniger's Hamiltonian for a gas of $N$ one-dimensional bosons on an interval of length $L$ (periodic b.c.), with a repulsive point interaction of strength $2c>0$, 
\begin{equation}
	\label{LLmodel}
	-\sum^N_{i=1}\partial^2_{x_i}+2c\sum_{1\leq i<j\leq N}\delta(x_i-x_j).
\end{equation}
The ground state can be found explicitly \cite{lieb1963exact}, and in the thermodynamic limit $L\to\infty$ with density $\rho=N/L$ fixed, its energy is
\begin{equation}
	\label{LLtherm}
	E_{\text{LL}}=N\rho^2 e(c/\rho),
\end{equation}
where $e(c/\rho)$ is described by integral equations.
Since $c/\rho$ is the only relevant parameter, diluteness, or low density $\rho$, should imply $c/\rho\gg1$. In this case, the ground state energy can be expanded as (\cite{lieb1963exact}; see, for example, \cite{guan2011polylogs,jiang2015understanding}),
\begin{equation}
	\label{LLenergy}
	E_{\text{LL}}=N\rho^2 e(c/\rho)=N\frac{\pi^2}{3}\rho^2\left(\left(1+2\frac{\rho}{c}\right)^{-2}+\mathcal{O}\left(\frac{\rho}{c}\right)^3\right).
\end{equation}
Recall that the dilute limit is $\rho a^2\ll1$ in 2D and $\rho a^3\ll1$ in 3D. This seems easy to generalize to 1D, but it turns out the Lieb--Liniger potential $2c\delta$ has scattering length $a=-2/c$. That is, in 1D the scattering length can be negative even if the potential is positive, and we should be careful to define the dilute limit as $\rho|a|\ll1$. This then matches the limit $c/\rho\gg1$ mentioned above, and we can write $\eqref{LLenergy}$ as
\begin{equation}
	\label{LLenergyina}
	\begin{aligned}
		E_{\text{LL}}&=N\frac{\pi^2}{3}\rho^2\left(\left(1-\rho a\right)^{-2}+\mathcal{O}(\rho |a|)^3\right)\\
		&=N\frac{\pi^2}{3}\rho^2\left(1+2\rho a+3(\rho a)^2+\mathcal{O}(\rho |a|)^3\right).
	\end{aligned}
\end{equation}
This expansion should be a good candidate for the 1D equivalent of \eqref{result3D} and \eqref{result2D}. This is supported by the fact that 1D bosons with a hard core of diameter $a$ have an exact thermodynamic ground state energy of \cite{girardeau1960relationship,astrakharchik2010low}
\begin{equation}
	\label{eqhardcore}
	N\frac{\pi^2}{3}\left(\frac{N}{L-Na}\right)^2=N\frac{\pi^2}{3}\rho^2\left(1-\rho a\right)^{-2}.
\end{equation}
This is the 1D free Fermi energy on an interval shortened by the space taken up by the hard cores (the ground state is of Girardeau type; see Remark \ref{remfermi} and the discussion of the Girardeau wave function in Section \ref{SecProofidea}). With two explicit examples satisfying \eqref{LLenergyina} to second order, it seems likely we can expect this expansion to be general \cite{astrakharchik2010low}, just like \eqref{result3D} and \eqref{result2D} in three and two dimensions. Indeed, our main result confirms the validity of \eqref{LLenergyina} to first order, for a wide class of interaction potentials. 

\subsection{Main theorem}
\label{secMain}
We will consider the $N$-body Hamiltonian 
\begin{equation}
	\label{H_N}
	H_N=-\sum^N_{i=1}\partial^2_{x_i}+\sum_{1\leq i<j\leq N}v(x_i-x_j)
\end{equation}
on the interval $[0,L]$ with some potential $v$ (see assumptions below) and any choice of (local, self-adjoint) boundary conditions. Let $\dom{H_N}$ be the appropriate bosonic domain of symmetric wave functions for this set-up. The ground state energy is then
\begin{equation}
	\label{vari}
	E(N,L):=\inf_{\substack{\Psi\in\dom{H_N}\\\|\Psi\|=1}}\bra{\Psi}H_N\ket{\Psi}=\inf_{\substack{\Psi\in\dom{H_N}\\\|\Psi\|=1}}\mathcal{E}(\Psi),
\end{equation}
with energy functional
\begin{equation}
	\label{enfunctional}
	\mathcal{E}(\Psi)=\int_{[0,L]^N}\sum_{i=1}^{N}\abs{\partial_i\Psi}^2+\sum_{1\leq i<j\leq N}v_{ij}\abs{\Psi}^2.
\end{equation}
In this paper, we will only work with the energy functional. This allows us to study rather general potentials. \\
\textbf{Assumptions on the potential.} Throughout the paper, we will assume that the 2-body potential $v$ is a symmetric Borel measure\footnote{In the case the measure is not $\sigma$-finite, we use the convention that$$\int \abs{\Psi}^2 v(x_i-x_j)\diff ^Nx=\int \abs{\Psi}^2\diff_{x_i-x_j=\textnormal{fixed}}^{N-1}x \diff v(x_i-x_j),$$ where $\diff_{x_i-x_j=\textnormal{fixed}}^{N-1}x$ is the Lebesgue measure with the property $$\diff^N x=\diff_{x_i-x_j=\textnormal{fixed}}^{N-1}x\times \diff (x_i-x_j).$$ In other words, the potential acts at the level of the two-particle reduced density.} with the decay property\footnote{Notice that we do not assume $v$ to be absolutely continuous with respect to the Lebesgue measure, but we abuse notation and write $v(x)dx$ when integrating with respect to $v$.}  $\int_{b}^\infty v(x)x^{6+\epsilon}\diff x<\infty$ for some $b\geq0$ and some $\epsilon>0$. In particular, hard-core potentials, delta potentials, and potentials with sufficiently fast polynomial decay are included. 

For potentials with compact support (say contained in $[-R_0,R_0]$), we can use $b=R_0$ as a natural length scale, and we will study $\rho R_0\ll 1$. In general, for potentials without compact support, the following length scales will appear in the error of our main result (see also Remark \ref{lengthscalesremark}). These vanish in the dilute limit $\rho|a|\ll0$ because of our assumption.
	\begin{equation}
 \label{scales}
		\begin{aligned}
   \varepsilon^u_v(\rho)&\coloneqq\min_{\frac{(\rho\abs{a})^{4/5}}{\rho}\leq b\leq  \frac{1}{16}\frac1\rho} \left(\left(\rho b\right)^{1/2}b+ \int_{b}^{\infty}v(x) x^2\diff x+\frac{1}{\rho^2 }\int_{b}^\infty v(x)\diff x\right)\\
			 \varepsilon^l_v(\rho)&\coloneqq\min_{b\geq \abs{a}} \left(\left(\rho b\right)^{1/5}b+\int_{b}^\infty v(x) x^{2}\diff x\right)\\\varepsilon_v(\rho)&\coloneqq\max(\varepsilon^u_v(\rho),\varepsilon^l_v(\rho))
		\end{aligned}
	\end{equation}
That is, $\varepsilon_v(\rho)\to 0$ as $\rho\to 0$ for any $v$ satisfying the assumptions.\footnote{
To see this for e.g.\ $\varepsilon^u_v(\rho)$, use $\int_{b}^{\infty} v(x)\diff x\leq \frac{1}{b^{6+\epsilon}}\int_{b}^{\infty} v(x) x^{6+\epsilon}\diff x$ and consider $b=R^{(4+\epsilon)/(6+\epsilon)}\rho^{-2/(6+\epsilon)}$, with $R$ some fixed length scale.}. Note the integrals can be avoided if the potential has compact support (see Remark \ref{lengthscalesremark}).

\begin{theorem}[bosons]
	\label{TheoremMain}
	Consider a Bose gas with repulsive interaction $v$ as defined above. Write $\rho=N/L$. For $\rho|a|$ and $\rho \varepsilon_v(\rho)$ sufficiently small, the ground state energy can be expanded as 
	\begin{equation}
		\label{result}
		E(N,L)=N\frac{\pi^2}{3}\rho^2\left(1+2\rho a+
		\mathcal{O}
		\left(\rho \varepsilon_v(\rho)+N^{-2/3}\right)\right),
	\end{equation}
	where $a$ is the scattering length of $v$ (see Lemma \ref{lemscatlength} below). A precise expression for the error is given in the upper and lower bounds \eqref{equpper} and \eqref{eqlower}.
\end{theorem}
To obtain this result, we prove an upper bound in the form of Proposition \ref{PropositionUpperBound} in Section \ref{SecUpperbound}, and a matching lower bound in the form of Proposition \ref{PropositionLowerBound} in Section \ref{SecLowerbound}. We use Dirichlet boundary conditions for the upper bound and Neumann boundary conditions for the lower bound, as these produce the highest and lowest ground state energy respectively. That way, Theorem \ref{TheoremMain} holds for a wide range of boundary conditions. 

\begin{remark}
\label{remfermi}
As a special case, Theorem \ref{TheoremMain} covers the ground state energy expansion \eqref{LLenergy} of the Lieb--Liniger model \eqref{LLmodel} in the limit $c/\rho\gg1$, as discussed in the introduction. This is known as the Tonks--Girardeau limit. Crucially, in this limit, the leading order term is the energy of the 1D free Fermi gas $N\pi^2\rho^2/3$, as first understood by Girardeau \cite{girardeau1960relationship} (see also the discussion around \eqref{fermisolution} and \eqref{girardeausolution} below).\footnote{Note that Girardeau studied the $c/\rho\to\infty$ case before Lieb and Liniger, who then generalized his work to obtain and solve the complete Lieb-Liniger model \eqref{LLmodel}.} Theorem \ref{TheoremMain} shows this holds for general potentials as well. That means that the dilute limit in 1D is very different from the one in two and three dimensions, where the zeroth-order term in the energy is that of a perfect condensate at zero momentum and the first-order term can be extracted using Bogoliubov theory \cite{bogoliubov1947theory}. In particular, the free Bose gas ($v=0$) in 1D cannot be considered dilute, because it has infinite $|a|$ (see also Remark \ref{4}).
\end{remark}

\begin{remark}
An interesting feature of Theorem \ref{TheoremMain} is that the scattering length $a$ can be both positive and negative. In this sense, our result covers cases that do not necessarily resemble the Lieb--Liniger model, which always has a negative scattering length. We discuss this further in Section \ref{SecConfinement}. 
\end{remark}
\begin{remark}
\label{lengthscalesremark}
Note that zero scattering length can be achieved, which means the error in \eqref{result} cannot just be written in terms of $(\rho|a|)^s$ for some $s>1$ (if so, there would be a whole class of positive potentials that nevertheless exactly give the free Fermi energy). That is why the error $\rho \varepsilon_v(\rho)$ in terms of some length scale $\varepsilon_v(\rho)$ cannot be avoided. Its exact shape \eqref{scales} follows from our proof ($\varepsilon^u_v(\rho)$ from the upper bound, $\varepsilon^l_v(\rho)$ from the lower bound). The errors can be simplified somewhat for potentials with compact support in $[-R_0,R_0]$: as long as $\rho R_0\ll1$, one could take $b=R_0$ in the minimizations (or potentially $b=(\rho\abs{a})^{4/5}/\rho$ for $\varepsilon^u_v(\rho)$) and avoid errors involving integrals of $v$.
\end{remark}

\subsection{Proof strategy}
\label{SecProofidea}
The most important ingredient in our proof is the following lemma, which follows from straightforward variational calculus. It is based on work by Dyson on the 3D Bose gas \cite{dyson1957ground} and is present in Appendix C in \cite{lieb2006mathematics}.
\begin{lemma}[The 2-body scattering solution and scattering length]
	\label{lemscatlength}
	Suppose $v$ is a symmetric Borel measure with compact support, so that $\supp(v)\subset[-R_0,R_0]$. Let $R\geq R_0$. For all $f\in H^1[-R,R]$ subject to $f(R)=f(-R)=1$,
	\begin{equation}
		\label{dyson1}
		\int^{R}_{-R}2|\partial_xf|^2+v(x)|f(x)|^2\diff x\geq \frac{4}{R-a}.
	\end{equation}
	There is a unique $f_0$ attaining the lower bound in \eqref{dyson1}: the scattering solution. It satisfies the scattering equation $\partial_x^2f_0=\frac12vf_0$ in the sense of distributions, and \begin{equation}
	    f_0(x)=(x-a)/(R-a),\qquad\textnormal{for }x\in[R_0,R].
	\end{equation} The parameter $a$ is called the scattering length (note it can be negative). 
\end{lemma}

The following lemma allows for non-compactly supported potentials. The proof can be found in Appendix \ref{AppendixA}, among with further facts and estimates on the scattering length. 
\begin{lemma}\label{lemscatlength2}
	Suppose $v\neq0$ is a symmetric Borel measure. Denote by $a_R$ the scattering length of $\mathbbm{1}_{[-R,R]}v$. Then, \begin{equation}
 \label{eqnincreasing}
	a_R\uparrow a\ \textnormal{ as }R\to\infty,
	\end{equation}
	for some $a$,  which we call the scattering length of $v$, with $a_R\leq a\leq\infty$. Furthermore, for $R>-a_{R}$ (which holds if $R$ is large enough by \eqref{eqnincreasing}),\begin{equation}
 \label{estcutoff}
		a-a_R\leq 2\int_R^{\infty} v(x) x^2\diff x.
	\end{equation}
\end{lemma}

\begin{remark}
\label{4}
    For $v=0$, the statements of Lemmas \ref{lemscatlength} and \ref{lemscatlength2} above hold with the interpretation that $a_R=a=-\infty$ for all $R>0$. This means that the free Bose gas does not satisfy the diluteness condition $\rho|a|\ll1$, and it does not have a direct relation to cases studied in this paper. 
\end{remark}

Results similar to Lemma \ref{lemscatlength} play an important role in the understanding of the ground state energy expansions \eqref{result3D} and \eqref{result2D} in higher dimensions \cite{dyson1957ground,lieb1998ground,lieb2001ground}, but there are a number of things we need to do differently. These relate to the fermionic behaviour of the bosons in the limit $\rho|a|\ll1$ (see Remark \ref{remfermi} above). 

What does this mean in practice? For the upper bound in Section \ref{SecUpperbound}, it suffices to find a suitable trial state by the variational principle \eqref{vari}. Good trial states for dilute bosons in 2D and 3D are close to a pure condensate, but in 1D the state will have to be close to the free Fermi ground state obtained in the limit $\rho|a|\to0$. To achieve this, we can rely on Girardeau's solution \cite{girardeau1960relationship} of the $c/\rho\to\infty$ limit of the Lieb--Liniger model. In that case, the delta function in \eqref{LLmodel} enforces a zero boundary condition whenever two bosons meet, so the bosons are impenetrable. The wave function is then found by minimizing the kinetic energy subject to this boundary condition. If we only consider the sector $0\leq x_1\leq\dots\leq x_N\leq L$ (which suffices by symmetry), this is exactly the free Fermi problem. For periodic boundary conditions on the interval $[0,L]$, the (unnormalized) free Fermi ground state is\footnote{This expression can be found by creating a Slater determinant of momentum eigenstates, and noting this is a Vandermonde determinant. See Section \ref{secfreefermi} for the calculation for Dirichlet boundary conditions.}
\begin{equation}
	\label{fermisolution}
	\Psi^{\text{per}}_F(x_1,\dots,x_N)=\prod_{1\leq i<j\leq N}\sin\left(\pi\frac{x_i-x_j}{L}\right).
\end{equation}
Of course, the ground state for impenetrable bosons should be symmetric rather than antisymmetric, and to correctly extend it beyond $0\leq x_1\leq\dots\leq x_N\leq L$, we need to remove the signs,
\begin{equation}
	\label{girardeausolution}
	\abs{\Psi^{\text{per}}_F}(x_1,\dots,x_N)=\prod_{1\leq i<j\leq N}\abs{\sin\left(\pi\frac{x_i-x_j}{L}\right)}.
\end{equation}
This is Girardeau's ground state for impenetrable bosons, and it still produces the free Fermi kinetic energy $N\pi^2/3\rho^2$ in the thermodynamic limit.\footnote{The wave functions $\Psi^{\text{per}}_F$ and $|\Psi^{\text{per}}_F|$ have the same energy and that is all we will need in this paper. However, their momentum distributions are very different, which is discussed further in Section \ref{SecOpenproblems}.\label{momremark}} 

Returning to the problem of finding a suitable trial state, \eqref{girardeausolution} should be a good departure point. To account for the effect of the interaction potential, we should modify the $\sin(\pi(x_i-x_j)/L)$ terms in \eqref{girardeausolution} on the scale set by $a$. Lemma \ref{lemscatlength}, and the scattering solution $f_0$, are designed to provide the right 2-body wave function in the presence of the potential, so it seems natural to replace the sine by
\begin{equation}
	\label{trialidea}
	\begin{cases*}
		f_0(x)\sin(\pi b/L)& \phantom{when} $|x|\leq b$\\
		\sin(\pi\abs{x}/L)& \phantom{when} $|x|> b$
	\end{cases*}
\end{equation}
on some suitable scale $\abs{a}\ll b\ll L$. This is the idea we rely upon for the upper bound proved in Section \ref{SecUpperbound}. 

For the lower bound in Section \ref{SecLowerbound}, we equally need to find a way to obtain the free Fermi energy to leading order. We use Lemma \ref{lemscatlength} in combination with the known expansion \eqref{LLenergy} for the Lieb--Liniger model. Choosing a suitable $R>R_0$, the idea is that \eqref{dyson1} can be written as 
\begin{equation}
	\label{eqidea}
	\int^{R}_{-R}2|\partial_x f|^2+v(x)|f(x)|^2\diff x\geq \frac{2}{R-a}\int (\delta_{R}(x)+\delta_{-{R}}(x))|f(x)|^2\diff x,
\end{equation}
thus lower bounding the kinetic and potential energy on $[-R,R]$ by a symmetric delta potential at radius $R$. Heuristically, we proceed by repeatedly applying \eqref{eqidea} to an $N$-body wave function $\Psi$ to obtain the symmetric delta potential for any neighbouring pairs of bosons. Then---crucially---we throw away the regions where $|x_{i+1}-x_i|\leq R$, which is inspired by a similar step in \cite{lieb2004one}. This produces a lower bound since $v$ is repulsive. With these regions removed, the two delta functions at radius $|x_{i+1}-x_i|= R$ collapse into a single delta at $|x_{i+1}-x_i|= 0$, with value $4/(R-a)$. This gives the Lieb--Liniger model on a reduced interval, evaluated on some wave function, which can then be lower bounded using the Lieb--Liniger ground state energy  \eqref{LLenergy} (appropriately corrected for finite $N$, and the loss of norm of $\Psi$ from the thrown-out regions). 

All this may seem rather radical, but the heuristics work out: starting with an interval of length $L$, we cut it back to length $L-(N-1)R$, so that the Lieb--Liniger expansion \eqref{LLenergy} with $c=2/(R-a)$ and new density $N/(L-(N-1)R)=\rho(1+\rho R+\dots)$ produce
\begin{equation}
	\label{heurist}
	N\frac{\pi^2}{3}\rho^2(1+2\rho R+\dots)(1-2\rho(R-a)+\dots)=N\frac{\pi^2}{3}\rho^2(1+2\rho a+\dots).
\end{equation}
We show that, a priori, the ground state wave function has little weight in the regions that get thrown out, so that \eqref{heurist} is accurate. The rigorous procedure used to obtain the Lieb--Liniger model and the expansion \eqref{heurist} are outlined in Section \ref{SecLowerbound}.

\subsection{Spinless fermions and anyons}
\label{SecOthersymmetries}
The expansion in Theorem \ref{TheoremMain} generalizes to spinless fermions in 1D. Given the antisymmetry of the fermionic wave function, the result involves the odd-wave scattering length $a_{\text{odd}}$ of $v$, obtained from Lemma \ref{lemscatlength} by replacing the symmetric boundary condition $f(R)=f(-R)=1$ by an antisymmetric one, $f(R)=-f(-R)=1$.
\begin{theorem}[spinless fermions]\label{TheoremFermion}
	Consider a Fermi gas with repulsive interaction  $v$ as defined before Theorem \ref{TheoremMain}. Let  $a_{\text{odd}}$ be the odd-wave scattering length of $v$. Define $\mathcal{D}_F(H_N)$ to be the appropriate domain of antisymmetric wave functions, and let $ E_F(N,L)$ be its associated ground state energy. Write $\rho=N/L$. For $\rho a_{\text{odd}}$ and $\rho \varepsilon_v(\rho)$ sufficiently small, the ground state energy can be expanded as 
	\begin{equation}
		E_F(N,L)=N\frac{\pi^2}{3}\rho^2\left(1+2\rho a_{\text{odd}}+\mathcal{O}\left(\rho \varepsilon_v(\rho)+N^{-2/3}\right)\right).
	\end{equation}
\end{theorem}

This theorem follows from Theorem \ref{TheoremMain} by using Girardeau's insight \cite{girardeau1960relationship} that fermions and impenetrable bosons in 1D are unitarily equivalent, and hence have the same energy. It suffices to know the wave function on a single sector $0\leq x_1\leq \dots\leq x_N\leq L$, after which we can extend to any other sector by adding the correct sign for either bosons or fermions (note that any acceptable wave function is zero whenever $x_i=x_j$). Flipping these signs is exactly the nature of the unitary operator; see for example the equivalence between \eqref{fermisolution} and \eqref{girardeausolution} discussed above. Given that Theorem \ref{TheoremMain} holds for impenetrable bosons, we can apply it as long as we use a zero boundary condition at $x=0$ in Lemma \ref{lemscatlength}. By similar reasoning, this produces the same scattering length as using the fermionic boundary condition $f(R)=-f(-R)=1$ in Lemma \ref{lemscatlength}. Theorem \ref{TheoremFermion} is therefore a corollary of Theorem \ref{TheoremMain}.

Inspired by our work, \cite{lauritsen2024ground} recently found the same upper bound for spinless fermions by means of a rigorous implementation of a fermionic cluster expansion.

\begin{remark}[spin-$1/2$ fermions]
	Consider the case of spin-$1/2$ fermions. If we study the usual, spin-independent Lieb--Liniger Hamiltonian \eqref{LLmodel}, the ground state will have a fixed total spin $S$. In fact, it is possible to study the ground state energy in each spin sector, and it will be monotone increasing in $S$ according to work by Lieb and Mattis \cite{lieb1962theory}. For each of these sectors, an explicit solution in terms of the Bethe ansatz exists \cite{yang1967some,gaudin1967systeme}. In certain cases, these can be expanded in the limit $c/\rho\gg1$ \cite{guan2011analytical}, and the analogue to \eqref{LLenergy} and \eqref{LLenergyina} can be obtained. The ground state energy for spin-$1/2$ fermions ($S=0$ by Lieb--Mattis) gives \cite{girardeau2006ground,guan2011analytical}
	\begin{equation}
		\label{spinexp}
		N\frac{\pi^2}{3}\rho^2\left(1-4\frac{\rho}{c}\ln(2)+\mathcal{O}(\rho/c)^2\right)=N\frac{\pi^2}{3}\rho^2\left(1+2\ln(2)\rho a+\mathcal{O}(\rho |a|)^2\right).
	\end{equation}
	Both the Lieb--Liniger exact solution and the expansions can be generalized to higher spins (or Young diagrams) \cite{sutherland1968further,guan2012one}. Note the leading order will be the free Fermi $N\pi^2\rho^2/3$ in all cases, since the delta potential does not influence the energy for impenetrable particles.
	
	For general potentials, the zeroth-order Fermi term is still expected to be correct, but the first-order term in \eqref{spinexp} has to be more complicated. Given that two spin-$1/2$ fermions can form symmetric and antisymmetric combinations, both the even-wave scattering length $a_{\text{even}}=a$ and the odd-wave scattering length $a_{\text{odd}}\geq 0$ of the potential will play a role. In the Lieb--Liniger example \eqref{spinexp}, $a_{\text{odd}}=0$, since the delta interaction does not affect antisymmetric wave functions. However, for hard-core fermions of diameter $a$, $a_{\text{odd}}=a_{\text{even}}=a$, and the energy should be \eqref{eqhardcore} since the spin symmetry plays no role. These two examples suggest that the correct formula is
	\begin{equation}
		N\frac{\pi^2}{3}\rho^2\big(1+2\ln(2)\rho a_{\text{even}}+2(1-\ln(2))\rho a_{\text{odd}}+\mathcal{O}(\rho \max(|a_{\text{even}}|,a_{\text{odd}}))^2\big).  
	\end{equation}
	We will demonstrate the validity of this expansion in a future publication.
\end{remark}

The approach followed to obtain Theorem \ref{TheoremFermion} can actually be taken further. What if, starting from some wave function on a sector $0\leq x_1<\dots<x_N\leq L$, we want to add anyonic phases $e^{i\kappa}$ with $0\leq \kappa\leq\pi$, whenever two particles are interchanged? It turns out this can be made to work, going back to, amongst others,  \cite{leinaas1977theory,kundu1999exact} (see \cite{posske2017second,bonkhoff2021bosonic} for a historical overview of this approach, comparisons with other versions of 1D anyonic statistics, and a discussion of experimental relevance). Just like fermions are unitarily equivalent to impenetrable bosons, these 1D anyons are equivalent to bosons with a certain choice of boundary conditions whenever two bosons meet. This can be related to the Lieb--Liniger model with some $c$ \cite{posske2017second}, since the delta function potential in \eqref{LLmodel} also imposes boundary conditions whenever two bosons meet. Hence, the (bosonic) Lieb--Liniger model can be viewed as a description of a non-interacting gas of anyons, with the $c/\rho\to\infty$ case being equivalent to fermions ($\kappa=\pi$) as understood by Girardeau. 

Somewhat confusingly, this does not complete the picture, because many authors study gases of 1D anyons themselves interacting through a Lieb--Liniger potential, see for example \cite{batchelor2006one,hao2008ground}. In this case, there are two parameters: the statistical parameter $\kappa$ describing the phase $e^{i\kappa}$ upon particle exchange, and the Lieb--Liniger parameter $c$. Not surprisingly, this set-up is again unitarily equivalent to the bosonic Lieb--Liniger model, with an interaction potential of $2c\delta_0/\cos(\kappa/2)$.\footnote{From the viewpoint of the energy, the combination $2c/\cos(\kappa/2)$ is the only relevant parameter. This is different for the momentum distribution, see Section \ref{SecOpenproblems}.} This means Theorem \ref{TheoremMain} can be applied. We provide more details about the set-up, and prove the following theorem as a corollary of Theorem \ref{TheoremMain}, in Section \ref{SectionOtherSymmetries}.

\begin{theorem}[anyons]
	\label{TheoremAnyon}
	Let $c\geq0$ and consider 1D anyons with statistical parameter $\kappa\in[0,\pi]$ and repulsive interaction $v=\tilde{v}+2c\delta_0$, where $\tilde{v}$ is a measure with $\tilde{v}(\{0\})=0$.
	Define $a_\kappa$ to be the scattering length associated with potential $ v_\kappa=\tilde{v}+\frac{2c}{\cos(\kappa/2)}\delta_0 $.
	Write $\rho=N/L$. For $\rho|a_\kappa|$ and $\rho \varepsilon_v(\rho)$ sufficiently small, the ground state energy $E_{(\kappa,c)}(N,L)$ of the anyon gas can be expanded as
	\begin{equation}
		E_{(\kappa,c)}(N,L)=N\frac{\pi^2}{3}\rho^2\left(1+2\rho a_{\kappa}+\mathcal{O}
		\left(\rho \varepsilon_v(\rho)+N^{-2/3}\right)\right).
	\end{equation}
\end{theorem}

\subsection{Physical applications and confinement from 3D to 1D}
\label{SecConfinement}
Given the general expansions \eqref{result3D} and \eqref{result2D} for the energy of dilute Bose gases in three and two dimensions, it is somewhat surprising that the possibility of a 1D equivalent was only hinted at in \cite{astrakharchik2010low}, and never studied in depth. The existence of the Lieb--Liniger model probably explains it: the model allow us to calculate all relevant quantities, and Lieb--Liniger physics naturally shows up in experimental settings in which 3D particles are confined to a 1D environment \cite{olshanii1998atomic,lieb2003one,lieb2004one,seiringer2008lieb}. However, our result does seem to add something that goes beyond the Lieb--Liniger model: it allows for positive scattering lengths $a$.

Mathematically, this seems clear. The scattering length of the Lieb--Liniger model with $c>0$ is $a=-2/c<0$, but Theorem \ref{TheoremMain} is also valid for potentials with a positive scattering length. There are plenty of interesting potentials with this property, and the energy shift has the opposite sign compared to the Lieb--Liniger case. (Note the Lieb--Liniger model with $c<0$ can be solved explicitly \cite{calabrese2007correlation}, but that it has a clustered ground state of energy $-\mathcal{O}(N^2)$ \cite{lieb1963exact,mcguire1964study}, so that scattering is irrelevant.)

Physically, the issue can seem more subtle. In the lab, 1D physics can be obtained by confining 3D particles with 3D potentials to a one-dimensional setting \cite{schreck2001quasipure,gorlitz2001realization,greiner2001exploring,moritz2003exciting}. As mentioned, the Lieb--Liniger model is very relevant to such set-ups \cite{olshanii1998atomic,lieb2003one,lieb2004one,seiringer2008lieb}, but only in certain parameter regimes. In these references, the confinement length $l_\perp$ in the trapping direction (a length necessarily small on some scale to create 1D physics) is much bigger than the range of atomic forces (or 3D scattering length). This allows excited states in the trapping direction to play a role in the problem, making the mathematical analysis complicated. The assumption that $l_\perp\gg a_{3D}$ is sometimes referred to as weak confinement \cite{bloch2008many}. 

There should also be a `strong confinement' regime $l_\perp\ll a_{3D}$, in which the excited states in the trapping direction play no role at all (presumably simplifying the mathematical steps needed to go from 3D to 1D). The problem would then essentially be 1D, and take on the form considered in Theorem \ref{TheoremMain}, thus allowing for positive 1D scattering lengths. We do not know whether the strong confinement regime is currently experimentally accessible.

\subsection{Open problems}
\label{SecOpenproblems}
\begin{enumerate}
	\item \textbf{The second-order term.}
	The second-order expansions \eqref{result3D} and \eqref{result2D} of the ground state energy of the dilute Bose gas in 3D and 2D hold for a wide class of potentials. As motivated in the introduction, the same might be true in the 1D expansion \eqref{LLenergyina}. However, techniques that are used in higher dimensions, such as (variants of) Bogoliubov theory, are not expected to be successful in 1D, given the ground state resembles Girardeau's gas of impenetrable bosons and is far from being a condensate or quasi-free bosonic state.
	
	\item \textbf{Momentum distribution.}
	As mentioned in Footnote \ref{momremark}, even though the 1D free Fermi ground state \eqref{fermisolution} and Girardeau's bosonic equivalent \eqref{girardeausolution} have the same energy, their momentum distributions differ. In the thermodynamic limit, the free Fermi ground state has a uniform momentum distribution, up to the Fermi momentum $|k|\leq k_F=\pi\rho$. Girardeau's state has the same quasi-momentum distribution, but the momentum distribution itself diverges like $1/\sqrt{k}$ for small $k$ \cite{lenard1964momentum,vaidya1979one}. At finite $N$, the $k=0$ occupation is $O(1)$ for fermions, while it is $O(\sqrt{N})$ for bosons.\\
	It is also possible to study the Lieb--Liniger ground state in this way \cite{colcelli2018deviations}. The bosonic zero-momentum occupation $\lambda_0$ in the limit $c/\rho\gg1$ is predicted to be 
	\begin{equation}
		\lambda_0\sim N^{\frac12+\frac{2\rho}{c}+\mathcal{O}(\rho/c)^2}=N^{\frac12-\rho a+\mathcal{O}(\rho |a|)^2}, 
	\end{equation}
	and one can ask if this holds for general potentials as well. The same question can be posed in the context of anyons \cite{colcelli2018deviations}, as the full prediction seems to be \cite{colcelli2018deviations,batchelor2006one}
	\begin{equation}
		\lambda_0\sim N^{\left(\frac12+\frac{2\rho}{c}\cos\left(\frac{\kappa}{2}\right)\right)\left(1-\left(\frac{\kappa}{\pi}\right)^2\right)+\mathcal{O}(\rho \cos(\kappa/2)/c)^2}=N^{\left(\frac12-\rho a_\kappa\right)\left(1-\left(\frac{\kappa}{\pi}\right)^2\right)+\mathcal{O}(\rho |a_\kappa|)^2}.
	\end{equation}
	\item \textbf{Positive temperature.} For $T>0$, one can again ask if quantities like the chemical potential and free energy only depend on $\rho a$ to lowest orders. Starting from the ideal Fermi gas and excluding volume as in the case of hard-core bosons (the equivalent of \eqref{eqhardcore}), it is possible to generate appropriate expressions that might be universal \cite{de2019beyond}. Proving these for a wide class of potentials is an open problem. 
\end{enumerate}


\section{Upper bound in Theorem \ref{TheoremMain}}	
\label{SecUpperbound} 
\begin{proposition}[Upper bound in Theorem \ref{TheoremMain}]
	\label{PropositionUpperBound}
	Consider a Bose gas with repulsive interaction, $v$, as defined above Theorem \ref{TheoremMain}, with Dirichlet boundary conditions. Write $\rho=N/L$. There exists a constant $C_\text{U}>0$  such that for $\rho|a|$, $\rho \varepsilon_v^u(\rho)\leq C_U^{-1}$, the ground state energy $E^D(N,L)$ satisfies
	\begin{equation}
		\label{equpper}
		E^D(N,L)\leq N\frac{\pi^2}{3}\rho^2\left(1+2\rho a + C_\text{U}\left(\rho \varepsilon_v^u(\rho)+N^{-1}\right)\right).
	\end{equation}
\end{proposition}

As explained in Section \ref{SecProofidea}, the proof relies on a trial state constructed from the free Fermi ground state. With Dirichlet boundary conditions, we cannot use $\abs{\Psi^{\text{per}}_F}$ from \eqref{girardeausolution}, and shall instead have to construct its Dirichlet equivalent, denoted by $\abs{\Psi_F}$ in this section. This will be done in Section \ref{secfreefermi}. Given a suitable scale $ b>0$ to be fixed later on, the (unnormalized) trial state is	
\begin{equation}
	\label{psiomega}
	\Psi_\omega(x)=\begin{cases}
		\omega(\rr(x))\frac{\abs{\Psi_F(x)}}{\rr(x)}& \text{if }\rr(x)<b\\
		\abs{\Psi_F(x)}&\text{if }\rr(x)\geq b,
	\end{cases}
\end{equation}  
where $ \omega(x)=f_0(x)b$ is constructed from the scattering solution $f_0$ from Lemma \ref{lemscatlength} with potential $\mathbbm{1}_{[-b,b]}v$ and $R=b$,  and $\rr(x):=\min_{i<j}(\abs{x_i-x_j}) $ is the distance between the closest pair of particles (note this is uniquely defined almost everywhere). In other words, we only modify $|\Psi_F|$ with the scattering solution for the closest pair. This is convenient for technical reasons, and it will turn out to suffice if the number of particles $N$ is not too big.

We will need another technical step: an argument that produces a trial state for arbitrary $N$ (and $L$) using the $\Psi_\omega$ defined in \eqref{psiomega}. This is done in Section \ref{secarbN} by dividing $[0,L]$ into small intervals, and patching copies of $\Psi_\omega$. 

First, we focus on the small-$N$ trial state $\Psi_\omega$. Our goal will be the following lemma.
\begin{lemma}
\label{LemmaUpperBoundFewParticles}
	Fix $b>\max(2a,0)$ and assume $N(\rho b)^3\leq 1$. Let $E_0=N\frac{\pi^2}{3}\rho^2(1+\mathcal{O}(1/N)) $ be the ground state energy of the Dirichlet free Fermi gas. The energy of the (unnormalized) trial state $\Psi_\omega$ defined in \eqref{psiomega} can be estimated as 
	\begin{equation}
		\begin{aligned}
			\mathcal{E}(\Psi_\omega)&:=\int_{[0,L]^N} \sum_{i=1}^{N}\abs{\partial_i\Psi_\omega}^2+\sum_{1\leq i<j\leq N}v_{ij}\abs{\Psi_\omega}^2\\
			&\leq E_0\Bigg[1+2\rho a\frac{b}{b-a}+\textnormal{const. } \left(N(\rho b)^3+\rho \int_b^\infty v(x) x^2+\rho\abs{a} \frac{\ln(N)}{N}\right)\Bigg].
		\end{aligned}
	\end{equation}
\end{lemma}
To prove this lemma, it is useful to divide the configuration space into various parts. For $i<j$, define 
\begin{equation}
	\begin{aligned}
		B&:=\{x\in\R^N\vert\ \mathcal{R}(x)<b \}\\
		A_{ij}&:=\{x\in\R^N\vert\ \abs{x_i-x_j}<b\}\subset B\\
		B_{ij}&:=\{x\in\R^N \vert\ \rr(x)<b,\ \rr(x)=\abs{x_i-x_j} \}\subset A_{ij}.
	\end{aligned}
\end{equation}
Note that $ \Psi_\omega$ equals $\abs{\Psi_F} $ on the complement of $B$, and that $ B_{ij} $ equals $ B $ intersected with the set $ \{\text{``particles $i$ and $j$ are closer than any other pair"}\} $. On the set $A_{12}$, we will use the shorthand $\Psi_{12}:=\omega(x_1-x_2)\frac{\Psi_F(x)}{(x_1-x_2)}$, and define the energies
\begin{equation}
	\begin{aligned}
		E_1&\coloneqq\binom{N}{2}\left[\int_{A_{12}}\left( \sum_{i=1}^{N}\abs{\partial_i\Psi_{12}}^2+v_{12}\abs{\Psi_{12}}^2-\sum_{i=1}^{N}\abs{\partial_i\Psi_F}^2\right)+\int_{B_{12}}\sum_{\substack{1\leq i<j\leq N\\(i,j)\neq (1,2)}}v_{ij}\abs{\Psi_{12}}^2\right], \\
		E_2^{(1)}&\coloneqq\binom{N}{2}2N\int_{A_{12}\cap A_{13}}\sum_{i=1}^{N}\abs{\partial_i\Psi_F}^2,\\ E_2^{(2)}&\coloneqq\binom{N}{2}\binom{N-2}{2}\int_{A_{12}\cap A_{34}}\sum_{i=1}^{N}\abs{\partial_i\Psi_F}^2,\\
		E_3&\coloneqq\binom{N}{2}\int_{A_{12}^\complement} v_{12}\abs{\Psi_F}^2.
	\end{aligned}
\end{equation}
Recall $E_0=N\frac{\pi^2}{3}\rho^2(1+\mathcal{O}(1/N)) $ is the ground state energy of the Dirichlet free Fermi gas. Roughly speaking, $E_1$ describes energy corrections due to two particles being closer than distance $b$, and will produce the first-order terms we are interested in. The term $E_3$ accounts for potentials with non-compact support, and may also contribute to the energy to relevant order. The terms $E_2^{(1)}$ and $E_2^{(2)}$ describe situations in which more than two particles are within distance $b$, but we will see these are negligible to the order we are interested in. 

The following estimate holds.
\begin{lemma}
\label{LemmaEnergyFunctionalBound}
	\begin{equation}\label{EqBound1}
		\mathcal{E}(\Psi_\omega)\leq E_0+ E_1+E_2^{(1)}+E_2^{(2)}+E_3.
	\end{equation}
\end{lemma}
The strategy to prove the upper bound for Theorem \ref{TheoremMain} (Proposition \ref{PropositionUpperBound}) is as follows. We first prove Lemma \ref{LemmaEnergyFunctionalBound} below. We then study the Dirichlet free Fermi ground state $\Psi_F$ in Section \ref{secfreefermi}, to prepare  the estimates of $E_1$, $E_2^{(1)}$, $E_2^{(2)}$ and $E_3$. We estimate $E_1$ in Section \ref{secE1}, $E^{(1)}_2$ and $E^{(2)}_2$ in Section \ref{secE2}, and $E_3$ in Section \ref{secE3}. Together, these prove Lemma \ref{LemmaUpperBoundFewParticles}, which will then be used to construct a suitable trial state for large in $N$ in Section \ref{secarbN}.
\begin{proof}[Proof of Lemma \ref{LemmaEnergyFunctionalBound}]
	Since $\Psi_\omega=\abs{\Psi_F} $ on the complement of $ B=\{x\in\R^N \vert \rr(x)<b \} $, we can write, using the diamagnetic inequality without magnetic potential, \footnote{Strictly speaking, the diamagnetic inequality is not needed, as the estimate can be shown to be an equality in this case.}
	\begin{equation}
		\mathcal{E}(\Psi_\omega)\leq E_0+\int_B \left[\sum_{i=1}^{N}\abs{\partial_i\Psi_\omega}^2+\sum_{1\leq i<j\leq N}v_{ij}\abs{\Psi_\omega}^2-\sum_{i=1}^{N}\abs{\partial_i\Psi_F}^2\right]+\int_{B^\complement}\sum_{1\leq i<j\leq N}v_{ij}\abs{\Psi_F}^2,
	\end{equation}
    	with $ E_0=N\frac{\pi^2}{3}\rho^2(1+\mathcal{O}(1/N)) $ the ground state energy of the Dirichlet free Fermi gas. Using symmetry under the exchange of particles and the fact that $B_{ij}\cap B_{kl}=\emptyset$ for $(i,j)\neq (k,l)$ and $(i,j)\neq (l,k)$, and using the diamagnetic inequality without magnetic potential in the first sum to obtain second inequality, we find \begin{equation}
		\begin{aligned}
			\mathcal{E}(\Psi_\omega)&\leq E_0+E_3+\binom{N}{2}\int_{B_{12}}\left[ \sum_{i=1}^{N}\abs{\partial_i\Psi_\omega}^2+\sum_{1\leq i<j\leq N}v_{ij}\abs{\Psi_\omega}^2-\sum_{i=1}^{N}\abs{\partial_i\Psi_F}^2\right]\\&
			\leq E_0+E_3+\binom{N}{2}\int_{B_{12}}\left[ \sum_{i=1}^{N}\abs{\partial_i\Psi_{12}}^2+\sum_{1\leq i<j\leq N}v_{ij}\abs{\Psi_{12}}^2-\sum_{i=1}^{N}\abs{\partial_i\Psi_F}^2\right].
		\end{aligned}
	\end{equation}
	Since we have $ v\geq0 $, it follows that
	\begin{equation}
		\begin{aligned}
			\mathcal{E}(\Psi_\omega)&\leq E_0+E_3\\
   &\ +\binom{N}{2}\left[\int_{A_{12}}\left( \sum_{i=1}^{N}\abs{\partial_i\Psi_{12}}^2+v_{12}\abs{\Psi_{12}}^2-\sum_{i=1}^{N}\abs{\partial_i\Psi_F}^2\right)+\int_{B_{12}}\sum_{\substack{1\leq i<j\leq N\\(i,j)\neq (1,2)}}v_{ij}\abs{\Psi_{12}}^2\right]\\&\ 
			-\binom{N}{2}\int_{A_{12}\setminus B_{12}} \left(\sum_{i=1}^{N}\abs{\partial_i\Psi_{12}}^2+v_{12}\abs{\Psi_{12}}^2-\sum_{i=1}^{N}\abs{\partial_i\Psi_F}^2\right)\\&
			\leq E_0+E_1+E_3+\binom{N}{2}\int_{A_{12}\setminus B_{12}}\sum_{i=1}^{N}\abs{\partial_i\Psi_F}^2.
		\end{aligned}
	\end{equation}	
	Noting that $x\in A_{12}\setminus B_{12}$ implies $x\in A_{ij}$ for some $(i,j)\neq (1,2)$, we may, by antisymmetry of $\Psi_F$, estimate\begin{equation}
		\begin{aligned}
			\int_{A_{12}\setminus B_{12}}\sum_{i=1}^{N}\abs{\partial_i\Psi_F}^2&\leq 2N\int_{(A_{12}\setminus B_{12})\cap A_{13}}\sum_{i=1}^{N}\abs{\partial_i\Psi_F}^2+\binom{N-2}{2}\int_{(A_{12}\setminus B_{12})\cap A_{34}}\sum_{i=1}^{N}\abs{\partial_i\Psi_F}^2\\
			&\leq 2N\int_{A_{12}\cap A_{13}}\sum_{i=1}^{N}\abs{\partial_i\Psi_F}^2+\binom{N-2}{2}\int_{A_{12}\cap A_{34}}\sum_{i=1}^{N}\abs{\partial_i\Psi_F}^2.
		\end{aligned}
	\end{equation}
	Thus we find $
	\mathcal{E}(\Psi_\omega)\leq E_0+E_1+E_2^{(1)}+E_2^{(2)}+E_3$ as desired.
\end{proof}
\subsection{The free Fermi ground state with Dirichlet b.c.}
\label{secfreefermi}
The Dirichlet eigenstates of the Laplacian are $ \phi_j(x)=\sqrt{2/L}\sin(\pi j x/L) $. Thus, the (normalized|) Dirichlet free Fermi ground state is \begin{equation}
	\begin{aligned}
	    \Psi_F(x)&=\frac{1}{\sqrt{N!}}\det\left(\phi_j(x_i)\right)_{i,j=1}^{N}\\
     &=\frac{1}{\sqrt{N!}}\sqrt{\frac{2}{L}}^N\left(\frac{1}{2i}\right)^N\begin{vmatrix}
		\euler{iy_1}-\euler{-iy_1}&\euler{i2y_1}-\euler{-i2y_1}&\ldots&\euler{iNy_1}-\euler{-iNy_1}\\
		\euler{iy_2}-\euler{-iy_2}&\euler{i2y_2}-\euler{-i2y_2}&\ldots&\euler{iNy_2}-\euler{-iNy_2}\\
		\vdots&\vdots&\ddots&\vdots\\
		\euler{iy_N}-\euler{-iy_N}&\euler{i2y_N}-\euler{-i2y_N}&\ldots&\euler{iNy_N}-\euler{-iNy_N}
	\end{vmatrix},
	\end{aligned}
\end{equation}
where we defined $ y_j=\frac{\pi}{L}x_j$. Defining $ z_j=\euler{iy_j} $ and using the relation $ (x^n-y^n)/(x-y)=\sum_{k=0}^{n-1}x^ky^{n-1-k} $, we find\begin{equation}
	\Psi_F(x)=\frac{1}{\sqrt{N!}}\sqrt{\frac{2}{L}}^N\left(\frac{1}{2i}\right)^N\prod_{i=1}^{N}(z_i-z_i^{-1})\begin{vmatrix}
		1&z_1+z_1^{-1}&\ldots&\sum_{k=0}^{N-1}z_1^{2k-N+1}\\
		1&z_2+z_2^{-1}&\ldots&\sum_{k=0}^{N-1}z_2^{2k-N+1}\\
		\vdots&\vdots&\ddots&\vdots\\
		1&z_N+z_N^{-1}&\ldots&\sum_{k=0}^{N-1}z_N^{2k-N+1}\\
	\end{vmatrix}.
\end{equation}
Notice that $ (z+z^{-1})^n=\sum_{k=0}^{n}\binom{n}{k}z^{2k-n} $.
For $1\leq i\leq \ceil{\frac{N-1}{2}}$, we add $ \left(\binom{N-1}{i}-\binom{N-1}{i-1}\right) $ times column $ N-2i $ to column $ N $. This does not change the determinant, so \begin{equation}
	\Psi_F(x)=\frac{1}{\sqrt{N!}}\sqrt{\frac{2}{L}}^N\left(\frac{1}{2i}\right)^N\prod_{i=1}^{N}(z_i-z_i^{-1})\begin{vmatrix}
		1&z_1+z_1^{-1}&\ldots&\sum_{k=0}^{N-2}z_1^{2k-N+2}&(z_1+z_1^{-1})^{N-1}\\
		1&z_2+z_2^{-1}&\ldots&\sum_{k=0}^{N-2}z_2^{2k-N+2}&(z_2+z_2^{-1})^{N-1}\\
		\vdots&\vdots&\ddots&\vdots&\vdots\\
		1&z_N+z_N^{-1}&\ldots&\sum_{k=0}^{N-2}z_N^{2k-N+2}&(z_N+z_N^{-1})^{N-1}\\
	\end{vmatrix}.
\end{equation}
For $1\leq i\leq \ceil{\frac{N-2}{2}} $, we add $ \left(\binom{N-2}{i}-\binom{N-2}{i-1}\right) $ times column $ N-1-2i $ to column $ N-1 $, and continue this process. That is, for $3\leq j\leq N$ and $1\leq i\leq \ceil{\frac{N-j}{2}} $, we add  $ \left(\binom{N-j}{i}-\binom{N-j}{i-1}\right) $ times column $ N-j+1-2i $ to column $ N-j+1 $. This gives \begin{equation}
	\Psi_F(x)=\frac{1}{\sqrt{N!}}\sqrt{\frac{2}{L}}^N\left(\frac{1}{2i}\right)^N\prod_{i=1}^{N}(z_i-z_i^{-1})\begin{vmatrix}
		1&z_1+z_1^{-1}&(z_1+z_1^{-1})^2&\ldots&(z_1+z_1^{-1})^{N-1}\\
		1&z_2+z_2^{-1}&(z_2+z_2^{-1})^2&\ldots&(z_2+z_2^{-1})^{N-1}\\
		\vdots&\vdots&\vdots&\ddots&\vdots\\
		1&z_N+z_N^{-1}&(z_N+z_N^{-1})^2&\ldots&(z_N+z_N^{-1})^{N-1}\\
	\end{vmatrix}.
\end{equation}
This is a Vandermonde determinant, and we conclude \begin{equation}
	\begin{aligned}
		\Psi_F(x)&=\frac{1}{\sqrt{N!}}\sqrt{\frac{2}{L}}^N\left(\frac{1}{2i}\right)^N\prod_{k=1}^{N}(z_k-z_k^{-1})\prod_{i<j}^{N}\left((z_i+z_i^{-1})-(z_j+z_j^{-1})\right)\\
		&=\frac{1}{\sqrt{N!}}2^{\binom{N}{2}}\sqrt{\frac{2}{L}}^N\prod_{k=1}^{N}\sin\left(\frac{\pi}{L}x_k\right)\prod_{i<j}^{N}\left[\cos\left(\frac{\pi}{L}x_i\right)-\cos\left(\frac{\pi}{L}x_j\right)\right]\\
		&=-\frac{1}{\sqrt{N!}}2^{\binom{N}{2}+1}\sqrt{\frac{2}{L}}^N\prod_{k=1}^{N}\sin\left(\frac{\pi}{L}x_k\right)\prod_{i<j}^{N}\sin\left(\frac{\pi(x_i-x_j)}{2L}\right)\sin\left(\frac{\pi(x_i+x_j)}{2L}\right)
		.
	\end{aligned}
\end{equation}

\subsubsection{$k$-body reduced density matrices and Wick's theorem}
Given a wave function $\Psi\in L^2([0,L]^N)$, its $k$-particle reduced density matrix is given by\begin{equation}
	\gamma_{\Psi}^{(k)}(x_1,...,x_k;y_1,...,y_k)=\frac{N!}{(N-k)!}\int_{[0,L]^{N-k}}\overline{\Psi(x_1,...,x_N)}\Psi(y_1,...,y_k,x_{k+1},x_N)\diff x_{k+1}\ldots\diff x_N.
\end{equation}
Similarly, we define the the $k$-particle reduced density by\begin{equation}
	\rho_{\Psi}^{(k)}(x_1,...,x_k)=\gamma_{\Psi}^{(k)}(x_1,...,x_k;x_1,...,x_k).
\end{equation}
We will frequently abbreviate $\gamma_{\Psi_F}^{(k)}$ as $\gamma^{(k)}$ and $\rho_{\Psi_F}^{(k)}$ as $\rho^{(k)}$.
For a quasi-free state, Wick's theorem states that the $k$-point function may be expressed solely in terms of sums of products of two-point functions, with appropriate signs (see e.g.\ \cite{solovej2007many}, Theorem 10.2). For the free Fermi ground state (which has a fixed particle number), it implies
\begin{equation}
	\label{Wickexpression}
	\gamma^{(k)}(x_1,...,x_k;y_1,...,y_k) =\begin{vmatrix}
		\gamma^{(1)}(x_1;y_1)&\gamma^{(1)}(x_1;y_2)&\cdots&\gamma^{(1)}(x_1;y_k)\\
		\gamma^{(1)}(x_2;y_1)&\gamma^{(1)}(x_2;y_2)&\cdots&\gamma^{(1)}(x_2;y_k)\\
		\vdots&\vdots &\ddots&\vdots\\
		\gamma^{(1)}(x_k;y_1)&\gamma^{(1)}(x_k;y_2)&\cdots&\gamma^{(1)}(x_k;y_k)
	\end{vmatrix}.
\end{equation}
We use this to compute $\rho^{(2)}$ below. Using Taylor expansion and the bound \eqref{derivsbound} on $\gamma^{(1)}$ discussed in the next subsection, it will also be used to bound various reduced densities and density matrices.

\subsubsection{1-body reduced density matrix}
The 1-particle reduced density matrix of the Dirichlet free Fermi ground state is
\begin{equation}
	\begin{aligned}
		&\gamma^{(1)}(x;y)=\frac{2}{L}\sum_{j=1}^{N}\sin\left(\frac{\pi}{L}jx\right)\sin\left(\frac{\pi}{L} jy\right)=\frac{\sin\left(\pi\left(\rho+\frac{1}{2L}\right)(x-y)\right)}{2L\sin\left(\frac{\pi}{2L}(x-y)\right)}-\frac{\sin\left(\pi\left(\rho+\frac{1}{2L}\right)(x+y)\right)}{2L\sin\left(\frac{\pi}{2L}(x+y)\right)}.
	\end{aligned}
\end{equation}
We can write $ \gamma^{(1)}(x;y)$, as well as its translation invariant part $\tilde{\gamma}^{(1)}(x;y) $, in terms of the Dirichlet kernel $ D_n(x)=\frac{1}{2\pi}\sum_{j=-n}^{n}\euler{ijx}=\frac{\sin((n+1/2)x)}{2\pi\sin(x/2)} $, \begin{equation}
	\begin{aligned}
		\label{gamma1tildegamma1}
		\gamma^{(1)}(x;y)&=\frac{\pi}{L}\left(D_{N}\left(\pi\frac{x-y}{L}\right)-D_{N}\left(\pi\frac{x+y}{L}\right)\right),\\
		\tilde{\gamma}^{(1)}(x;y)&\coloneqq \frac{\pi}{L}D_{N}\left(\pi \frac{x-y}{L}\right).
	\end{aligned}
\end{equation}
A consequence is that 
\begin{equation}
	\label{derivsbound}
	\abs{\partial_{x}^{k_1}\partial_{y}^{k_2}\gamma^{(1)}(x;y)}\leq \frac{1}{\pi}(2N)^{k_1+k_2+1}\left(\frac{\pi}{L}\right)^{k_1+k_2+1}=\pi^{k_1+k_2}(2\rho)^{k_1+k_2+1}.
\end{equation}
Combined with Wick's theorem, discussed in the previous subsection, \eqref{derivsbound} implies bounds on (derivatives of) higher-order reduced density matrices of the free Fermi ground state that are uniform in all coordinates. Note the relevant power of $ \rho $ can be obtained directly from dimensional analysis. The bounds will be used later on to do Taylor expansions.\\
Other useful bounds that will be used in the proof of Lemma \ref{Lemma rho2 bound} are
\begin{equation}\label{EqDirichletKernelDerivativeBound}
	\begin{aligned}
		\int_{[0,L]}\abs{\rho^{(1)\prime}}&\leq \text{const. }\rho\ln(N),\\
		\int_{[0,L]}\abs{\rho^{(1)\prime\prime}}&\leq \text{const. }\rho^2\ln(N),
	\end{aligned}
\end{equation}
which follow from the textbook bound on the $L^1$-norm of the $m$th derivative of the Dirichlet kernel $$ \norm{\partial^mD_N}_{L^1[0,2\pi]}\leq \text{const. }N^{m}\ln(N). $$

\subsubsection{Useful bounds on various reduced density matrices of $\Psi_F$}
\begin{lemma}\label{Lemma rho2 bound}
	For the 2-body reduced density $\rho^{(2)}$ of the free Fermi ground state, it holds that
	\begin{equation}
		\rho^{(2)}(x_1,x_2)=\left(\frac{\pi^2}{3}\rho^4+f(x_2)\right)(x_1-x_2)^2+\mathcal{O}(\rho^6(x_1-x_2)^4), 
	\end{equation} with $ \int_{[0,L]}\abs{f(x_2)}\diff x_2\leq \textnormal{ const. }\rho^3\ln(N) $.
\end{lemma}
\begin{proof}
	Note that by translation invariance, we may Taylor expand $\tilde{\gamma}^{(1)}(x;y)$, defined in \eqref{gamma1tildegamma1}, in $x-y$ around $0$. Only even terms can appear as $D_N$ is even. Using \eqref{derivsbound}, we find \begin{equation}\label{EqGammaTaylorExpansion}
		\tilde{\gamma}^{(1)}(x;y)-(\rho+1/(2L))=\frac{\pi^2}{6}(\rho^3+\rho^2\mathcal{O}(1/L))(x-y)^2+\mathcal{O}(\rho^5(x-y)^4).
	\end{equation}
	Furthermore, it is easy to check that $ \gamma^{(1)}(x_1;x_2)-\rho^{(1)}\left((x_1+x_2)/2\right)=\tilde{\gamma}^{(1)}(x_1;x_2)-(\rho+1/(2L)) $. Now, by Wick's theorem \eqref{Wickexpression},  \begin{equation}
		\rho^{(2)}(x_1,x_2)=\rho^{(1)}(x_1)\rho^{(1)}(x_2)-\gamma^{(1)}(x_1;x_2)\gamma^{(1)}(x_2;x_1).
	\end{equation}
	Note that by Taylor's theorem and \eqref{derivsbound},
	\begin{equation}
		\begin{aligned}
			\rho^{(1)}(x_1)=\rho^{(1)}((x_1+x_2)/2)&+\rho^{(1)\prime}((x_1+x_2)/2)\frac{x_1-x_2}{2}\\&+\frac{1}{2}\rho^{(1)\prime\prime}((x_1+x_2)/2)\left(\frac{x_1-x_2}{2}\right)^2+\mathcal{O}(\rho^4(x_1-x_2)^3),
		\end{aligned}
	\end{equation}
	\begin{equation}
		\begin{aligned}
			\rho^{(1)}(x_2)=\rho^{(1)}((x_1+x_2)/2)&+\rho^{(1)\prime}((x_1+x_2)/2)\frac{x_2-x_1}{2}\\&+\frac{1}{2}\rho^{(1)\prime\prime}((x_1+x_2)/2)\left(\frac{x_1-x_2}{2}\right)^2+\mathcal{O}(\rho^4(x_1-x_2)^3),
		\end{aligned}
	\end{equation}
	where both expressions can be expanded further if needed. Using that $ \gamma^{(1)} $ is symmetric in its coordinates, we conclude from the previous three equations that \begin{equation}
		\begin{aligned}
			\rho^{(2)}(x_1,x_2)=\rho^{(1)}((x_1+x_2)/2)^2-\gamma^{(1)}(x_1;x_2)^2-\left[\rho^{(1)\prime}((x_1+x_2)/2)\right]^2\left(\frac{x_1-x_2}{2}\right)^2\\+\rho^{(1)}((x_1+x_2)/2)\rho^{(1)\prime\prime}((x_1+x_2)/2)\left(\frac{x_1-x_2}{2}\right)^2+\mathcal{O}(\rho^6(x_1-x_2)^4).
		\end{aligned}
	\end{equation}
	Terms of order $ \mathcal{O}(\rho^5(x_1-x_2)^3) $ cancel due to symmetry. Now, notice that \hbox{$ 0\leq\rho^{(1)}\leq 2\rho $} and \hbox{$\abs{\rho^{(1)'}}\leq 8\pi \rho^2 $} by \eqref{derivsbound}. Together with \eqref{EqDirichletKernelDerivativeBound}, this implies that
	\begin{equation}
		\begin{aligned}
			\label{firsteqlem}
			\rho^{(2)}(x_1,x_2)=\rho^{(1)}((x_1+x_2)/2)^2-\gamma^{(1)}(x_1;x_2)^2+g_1(x_1+x_2)(x_1-x_2)^2+\mathcal{O}(\rho^6(x_1-x_2)^4),
		\end{aligned}
	\end{equation}
	for some function $ g_1 $ satisfying $ \int_{[0,L]}\abs{g_1}\leq \text{const. }\rho^3\ln(N)$.
	Furthermore, notice that by \eqref{EqGammaTaylorExpansion} and the remark below it, we have 
	\begin{equation}
		\begin{aligned}
			\label{secondeqlem}
			&\rho^{(1)}((x_1+x_2)/2)^2-\gamma^{(1)}(x_1;x_2)^2\\
			&\hspace{1cm}=(\rho^{(1)}((x_1+x_2)/2)-\gamma^{(1)}(x_1;x_2))(\rho^{(1)}((x_1+x_2)/2)+\gamma^{(1)}(x_1;x_2))\\&\hspace{1cm}
			=\left[\rho+1/(2L)-\tilde{\gamma}^{(1)}(x_1;x_2)\right]\left[-\rho-1/(2L)+\tilde{\gamma}^{(1)}(x_1;x_2)+2\rho^{(1)}((x_1+x_2)/2)\right]\\&\hspace{1cm}
			=-\left[\rho+1/(2L)-\tilde{\gamma}^{(1)}(x_1;x_2)\right]^2+2\left[\rho+1/(2L)-\tilde{\gamma}^{(1)}(x_1;x_2)\right]\rho^{(1)}((x_1+x_2)/2)\\&\hspace{1cm}
			= 2\left(\frac{\pi^2}{6}(\rho^3+\rho^2\mathcal{O}(1/L))(x_1-x_2)^2+\mathcal{O}(\rho^5(x_1-x_2)^4)\right)\left(\rho+\frac{1}{2L}-\frac{\pi}{L}D_{N}((x_1+x_2)/(2L))\right)\\&\hspace{1cm}
			=\frac{\pi^2}{3}\rho^4(x_1-x_2)^2+g_2(x_1-x_2)(x_1-x_2)^2+\mathcal{O}(\rho^6(x_1-x_2)^4),
		\end{aligned}
	\end{equation}
	where we have chosen $ g_2(x)=\frac{\pi^2}{3}\rho^3\left(\frac{\text{const.}}{2L}+\abs{\frac{\pi}{L}D_N(x/(2L))} \right) $ which clearly satisfies $  \int_{[0,L]} g_2\leq \text{const. } \rho^3 \ln(N) $. Combining \eqref{firsteqlem} and \eqref{secondeqlem} now proves the lemma. 
\end{proof}
For a function $f(x)$, we denote $f(x)|_{x=x_1}=f(x_1)$ and $[f(x)]^{x_2}_{x=x_1}=f(x_2)-f(x_1)$. The following bounds will be useful. 
\begin{lemma}\label{LemmaDensityBounds}
	\begin{equation}
		\begin{aligned}
            \rho^{(2)}(x_1,x_2)&\leq 8\pi^2\rho^4(x_1-x_2)^2,\\
			\rho^{(3)}(x_1,x_2,x_3)&\leq \textnormal{const. }\rho^7(x_1-x_2)^2(x_3-x_2)^2,\\
			\rho^{(4)}(x_1,x_2,x_3,x_4)&\leq \textnormal{const. }\rho^8(x_1-x_2)^2(x_3-x_4)^2,\\
			\abs{\sum_{i=1}^{2}\partial_{y_i}^2\gamma^{(2)}(x_1,x_2;y_1,y_2)\rvert_{y=x}}&\leq \textnormal{const. } \rho^{6}(x_1-x_2)^2,\\
			\abs{\partial_{y_1}^2\left(\frac{\gamma^{(2)}(x_1,x_2;y_1,y_2)}{y_1-y_2}\right)\Bigg\rvert_{y=x}}&\leq \textnormal{const. } \rho^{6}\abs{x_1-x_2},\\
			\abs{\sum_{i=1}^{2}(-1)^{i-1}\partial_{y_i}\left(\frac{\gamma^{(2)}(x_1,x_2;y_1,y_2)}{y_1-y_2}\right)\Bigg\rvert_{y=x}}&\leq \textnormal{const. } \rho^{6}(x_1-x_2)^2,\\
			\sum_{i=1}^{3}\left(\partial_{x_i}\partial_{y_i}\gamma^{(3)}(x_1,x_2,x_3;y_1,y_2,y_3)\right)\Bigg\vert_{y=x}&\leq \textnormal{const. }\rho^9(x_1-x_2)^2(x_3-x_2)^2,\\
			\abs{\sum_{i=1}^{3}\left(\partial_{y_i}^2\gamma^{(3)}(x_1,x_2,x_3;y_1,y_2,y_3)\right)\Bigg\vert_{y=x}}&\leq\textnormal{const. }\rho^9(x_1-x_2)^2(x_3-x_2)^2,\\
			\abs{\left[\partial_{y}\gamma^{(4)}(x_1,x_2,x_3,x_4;y,x_2,x_3,x_4)\bigg\vert_{y=x_1}\right]_{x_1=x_2-b}^{x_2+b}}&\leq \textnormal{const. }\rho^8b(x_3-x_4)^2.
		\end{aligned}
	\end{equation}
\end{lemma}
\begin{proof}
	The bounds follow straightforwardly from Taylor's theorem and the symmetries of the left-hand sides.
	We give, in the following, two examples which we find to be representative for the general strategy.\\
	\textit{1.} Consider $ \sum_{i=1}^{2}\partial_{y_i}^2\gamma^{(2)}(x_1,x_2;y_1,y_2)\rvert_{y=x} $. Notice first that $ \sum_{i=1}^{2}\partial_{y_i}^2\gamma^{(2)}(x_1,x_2;y_1,y_2) $ is antisymmetric in $ (x_1,x_2) $ and in $ (y_1,y_2) $. As discussed after \eqref{derivsbound}, all derivatives of $ \gamma^{(k)} $ are bounded uniformly in its coordinates by a constant times $ \rho^k $ for some $ k\in\mathbb{N} $, so we can Taylor expand $ \partial^2\gamma^{(2)} $. By expanding $ x_1 $ around $ x_2 $ and $ y_1 $ around $ y_2 $, we see that antisymmetry implies $ \sum_{i=1}^{2}\partial_{y_i}^2\gamma^{(2)}(x_1,x_2;y_1,y_2)\leq \text{const. }\rho^6(x_1-x_2)(y_1-y_2) $, where the power of $ \rho $ can be found by dimensional analysis.\\
	\textit{2.
 } Consider $\abs{\sum_{i=1}^{2}(-1)^{i-1}\partial_{y_i}\left(\frac{\gamma^{(2)}(x_1,x_2;y_1,y_2)}{y_1-y_2}\right)\Bigg\rvert_{y=x}}$. We start by defining the coordinates $z_x\coloneqq (x_1-x_2)/2$, $z_x'\coloneqq (x_1+x_2)/2$, $z_y\coloneqq (y_1-y_2)/2$, and $z_y'\coloneqq (y_1+y_2)/2$. Furthermore, define $\hat{\gamma}^{(2)}(z_x,z_x';z_y,z_y')\coloneqq\gamma^{(2)}(z_x+z_x',z_x'-z_x;z_y+z_y',z_y'-z_y)$. By the antisymmetry of $\gamma^{(2)}$ in $x_1,x_2$ and $y_1,y_2$, we see that $\hat{\gamma}^{(2)}$ is odd in $z_x$ and $z_y$.
The change of variables means
$$
\begin{aligned}
2(\partial_{y_1}-\partial_{y_2})\left(\frac{\gamma^{(2)}(x_1,x_2;y_1,y_2)}{y_1-y_2}\right)&=
	\partial_{z_y}\left(\frac{\hat{\gamma}^{(2)}(z_x,z_x';z_y,z_y')}{z_y}\right)\\
 &=\frac{z_y\partial_{z_y}\hat{\gamma}^{(2)}(z_x,z_x';z_y,z_y')-\hat{\gamma}^{(2)}(z_x,z_x';z_y,z_y')}{z_y^2}.
 \end{aligned}
	$$
	Taylor expanding both terms in the numerator in $z_y$ and $z_x$ around $0$ to order $z_xz_y^3$ gives
	\begin{equation}
		\begin{aligned}
			\Bigg\rvert z_xz_y\partial_{z_x}\left[\partial_{z_y}\hat{\gamma}^{(2)}(z_x,z_x';z_y,z_y')\right]\Big\vert_{z_x=z_y=0}+\tfrac{1}{2}z_xz_y^3\partial_{z_x}\partial_{z_y}^2\left[\partial_{z_y}\hat{\gamma}^{(2)}(z_x,z_x';z_y,z_y')\right]\Big\vert_{z_x=z_y=0}\\-z_xz_y\partial_{z_x}\partial_{z_y}\left[\hat{\gamma}^{(2)}(z_x,z_x';z_y,z_y')\right]\Big\vert_{z_x=z_y=0}-\tfrac{1}{6}z_xz_y^3\partial_{z_x}\partial_{z_y}^3\left[\hat{\gamma}^{(2)}(z_x,z_x';z_y,z_y')\right]\Big\vert_{z_x=z_y=0}\\
			+\mathcal{O}\left(\rho^{8}(z_xz_y^5+z_x^3z_y^3)\right)\Bigg\lvert\\
			\leq \textnormal{const. }\rho^6\abs{z_xz_y^3},
		\end{aligned}
	\end{equation}
	where we used that $\hat{\gamma}^{(2)}(z_x,z_x';z_y,z_y')$ is odd in $z_x$ and $z_y$, to conclude that all even order terms vanish when Taylor expanding in these variables around $0$.
	The desired result follows.\\
\end{proof}

\subsection{Estimating $ E_1 $}
\label{secE1}
Recall $A_{12}=\{x\in\R^N\vert \abs{x_1-x_2}<b\}$ and $\Psi_{12}(x)=\frac{\omega(x_1-x_2)}{(x_1-x_2)}\Psi_F(x)$, as well as \begin{equation}
	E_1\coloneqq\binom{N}{2}\left[\int_{A_{12}}\left( \sum_{i=1}^{N}\abs{\partial_i\Psi_{12}}^2+v_{12}\abs{\Psi_{12}}^2-\sum_{i=1}^{N}\abs{\partial_i\Psi_F}^2\right)+\int_{B_{12}}\sum_{\substack{1\leq i<j\leq N\\(i,j)\neq (1,2)}}v_{ij}\abs{\Psi_{12}}^2\right].
\end{equation}
We prove the following bound. 
\begin{lemma}\label{LemmaE1Bound}
	For $b>\max(2a,0)$ and $\rho b<1$ we have
	\begin{equation}
 \label{eqnlemma12}
		E_1\leq E_0 \left(2\rho a\frac{b}{b-a}+ \textnormal{const.}\left(\ N(\rho b)^3\left[ 1+\rho\int_{x>b} v(x)x^2\right]+\rho \abs{a} \frac{\ln(N)}{N}\right)\right).
	\end{equation}
\end{lemma}
\begin{proof}
Note the $\omega$ used to define the trial state is the scattering solution for the restricted potential $\mathbbm{1}_{[-b,b]}v$. This means that some of the calculations in this proof result in the restricted scattering length $a_b$ instead of the full scattering length $a$, as discussed in Lemma \ref{lemscatlength2}. Since this lemma also shows that $a_b\leq a\leq a_b+2\int_{x>b} v(x)x^2$, it is always the case that we can reformulate everything in terms of $a$, at the cost of an error. In this proof, however, we notice that all upper bounds containing $a_b$ are increasing in $a_b$, which means we can exchange $a_b$ for $a\geq a_b$ at the very end of the proof without additional errors.

	We estimate $ E_1 $ by splitting it into four terms $ E_1=E_1^{(1)}+E_1^{(2)}+E_1^{(3)}+E_1^{(4)} $, with
	\begin{equation}
 \label{someeq}
		\begin{aligned}
			E_1^{(1)}&\coloneqq \binom{N}{2}\left[\int_{A_{12}}2\abs{\partial_1\Psi_{12}}^2+\int_{A_{12}} v_{12}\abs{\Psi_{12}}^2\right],\\
			E_1^{(2)}&\coloneqq-\binom{N}{2}\int_{A_{12}}\left(2\abs{\partial_1\Psi_F}^2+\sum_{i=3}^{N}\abs{\partial_i\Psi_F}^2\right),\\
			E^{(3)}_1&\coloneqq\binom{N}{2}\left(\int_{B_{12}}\sum_{\substack{1\leq i< j\leq N\\(i,j)\neq (1,2)}}v_{ij}\abs{\Psi_{12}}^2\right),\\
			E_1^{(4)}&\coloneqq\binom{N}{2}\int_{A_{12}}\sum_{i=3}^{N}\abs{\partial_i\Psi_{12}}^2.
		\end{aligned}
	\end{equation}
	By straightforward computation we find that \begin{equation}\label{EqE1(1)}
		\begin{aligned}
			E_1^{(1)}=\ &2\binom{N}{2}\int_{A_{12}}\Bigg[\abs{\partial_{1}\omega(x_1-x_2)}^2\abs{\frac{\Psi_F(x)}{x_1-x_2}}^2\\&+2\omega(x_1-x_2)\partial_1\omega(x_1-x_2)\frac{\Psi_F(x)}{x_1-x_2}\partial_{1}\left(\frac{\Psi_F(x)}{x_1-x_2}\right)\\&+\omega(x_1-x_2)^2\left(\partial_{1}\frac{\Psi_F(x)}{x_1-x_2}\right)^2+\frac12 v_{12}\abs{\Psi_{12}}^2\Bigg].
		\end{aligned}
	\end{equation} 
	Taking half of the second term, and applying partial integration to the first $\partial_1$ gives
	\begin{equation}
 \label{someeq323}
		\begin{aligned}
			\int_{A_{12}}&\omega(x_1-x_2)\partial_1\omega(x_1-x_2)\frac{\Psi_F(x)}{x_1-x_2}\partial_{1}\left(\frac{\Psi_F(x)}{x_1-x_2}\right)\\=&-\int_{A_{12}}\Bigg[\omega(x_1-x_2)\partial_1\omega(x_1-x_2)\frac{\Psi_F(x)}{x_1-x_2}\partial_{1}\left(\frac{\Psi_F(x)}{x_1-x_2}\right)\\&+\omega(x_1-x_2)^2\left(\partial_{1}\frac{\Psi_F(x)}{x_1-x_2}\right)^2+\omega(x_1-x_2)^2\frac{\Psi_F(x)}{x_1-x_2}\partial_{1}^2\left(\frac{\Psi_F(x)}{x_1-x_2}\right)\Bigg]\\
			&+\int \left[\omega(x_1-x_2)^2\frac{\Psi_F(x)}{x_1-x_2}\partial_{1}\left(\frac{\Psi_F(x)}{x_1-x_2}\right)\right]_{x_1=x_2-b}^{x_2+b} \diff x_2...\diff x_N.
		\end{aligned}
	\end{equation}
By using the product rule in the second integral of \eqref{someeq323} and integrating over the variables $x_3,\dots,x_N$, we find that \eqref{EqE1(1)} becomes 
\begin{equation}\label{EqE1(2)}
	\begin{aligned}
		&E_1^{(1)}=\\ &2\binom{N}{2}\int_{A_{12}}\Bigg[\abs{\partial_{1}\omega(x_1-x_2)}^2\abs{\frac{\Psi_F(x)}{x_1-x_2}}^2\\&-\omega(x_1-x_2)^2\frac{\Psi_F(x)}{x_1-x_2}\partial_{1}^2\left(\frac{\Psi_F(x)}{x_1-x_2}\right)+\frac12 v_{12}\abs{\Psi_{12}}^2\Bigg]\\
		&+\int\left[\left(\frac{\omega(x_1-x_2)}{x_1-x_2}\right)^2\left[\partial_{1}\left(\gamma^{(2)}(x_1,x_2;y,x_2)\right)\bigg\vert_{y=x_1}-\frac{\rho^{(2)}(x_1,x_2)}{x_1-x_2}\right]\right]_{x_1=x_2-b}^{x_2+b} \diff x_2.\end{aligned}
\end{equation} 
	For the first term in the second integral in \eqref{EqE1(2)}, we define (also see the notation discussed above Lemma \ref{LemmaDensityBounds})  \begin{equation}\label{EqGammaDeriv2.}
		\begin{aligned}
			\kappa_1&:=\int\left[\left(\frac{\omega(x_1-x_2)}{x_1-x_2}\right)^2\partial_{1}\left(\gamma^{(2)}(x_1,x_2;y,x_2)\right)\bigg\vert_{y=x_1}\right]_{x_1=x_2-b}^{x_2+b}\diff x_2\\
			&=\int\left[\partial_{1}\left(\gamma^{(2)}(x_1,x_2;y,x_2)\right)\bigg\vert_{y=x_1}\right]_{x_1=x_2-b}^{x_2+b}\diff x_2,
		\end{aligned}
	\end{equation}
Hence, we find \begin{equation}\label{EqE1(3)}
	\begin{aligned}
		&E_1^{(1)}=\\ &\int_{\abs{x_1-x_2}<b}\Bigg[\left(\abs{\partial_{1}\omega(x_1-x_2)}^2+\frac12 v_{12}\abs{\omega(x_1-x_2)}^2\right)\frac{\rho^{(2)}(x_1,x_2)}{(x_1-x_2)^2}\\&-\omega(x_1-x_2)^2\frac{1}{x_1-x_2}\partial_{y_1}^2\left(\frac{\gamma^{(2)}(x_1,x_2;y_1,y_2)}{y_1-y_2}\right)\Bigg\rvert_{y=x}\Bigg]\\
		&+\kappa_1-\int\left[\frac{\rho^{(2)}(x_1,x_2)}{x_1-x_2}\right]_{x_1=x_2-b}^{x_2+b} \diff x_2	\end{aligned}
\end{equation}
	Thus, using Lemma \ref{lemscatlength} in the first line,  Lemma \ref{Lemma rho2 bound} in the first and last lines, and applying the fifth statement of Lemma \ref{LemmaDensityBounds} in the second line, we find \begin{equation}
		\begin{aligned}
			E_1^{(1)}&\leq\frac{\pi^2}{3}N\rho^2\left(2\rho\left[\frac{b^2}{b-a_b}-b\right]\left(1+\text{const. }\frac{\ln(N)}{N}\right)+\text{const. }(\rho b)^3\right)+\kappa_1 \\
			&\leq\frac{\pi^2}{3}N\rho^2\left( 2\rho a_b\frac{b}{b-a_b}+\text{const. }\left(N(\rho b)^3+\rho \abs{a}\frac{\ln(N)}{N}\right)\right)+\kappa_1.
		\end{aligned}
	\end{equation}
	For $E_1^{(2)}$, starting from the definition, using antisymmetry of $\Psi_F$, and applying integration by parts (note only the variable $x_1$ gives a boundary term), we find \begin{equation}
		\begin{aligned}
			E_1^{(2)}&=-\binom{N}{2}\int_{A_{12}}\left(2\abs{\partial_1\Psi_F}^2+\sum_{i=3}^{N}\abs{\partial_i\Psi_F}^2\right)\\&=-\binom{N}{2}\int_{A_{12}}\sum_{i=1}^{N}\overline{\Psi_F}(-\partial^2_i\Psi_F)-2\binom{N}{2}\int\left[\overline{\Psi_F}\partial_1\Psi_F\right]_{x_1=x_2-b}^{x_1=x_2+b}\\
			&=-E_0\binom{N}{2}\int_{A_{12}}\abs{\Psi_F}^2-\underbrace{\int\left[\partial_y\gamma^{(2)}(x_1,x_2;y,x_2)\vert_{y=x_1}\right]_{x_2-b}^{x_2+b} \diff x_2}_{\kappa_1}\\
			&\leq -\kappa_1,
		\end{aligned}
	\end{equation}
 where we used that $\Psi_F$ is an eigenfunction of $-\sum_{i=1}^{N}\partial_i^2$ with eigenvalue $E_0$ for the third equality.
	The error $E_1^{(3)}$ can be estimated as follows.
	Notice that on $B_{12}$ we have $ \abs{x_i-x_j}\geq\abs{x_1-x_2}$ for any $i\neq j$. Since $\omega$ is increasing, this implies $\omega(x_i-x_j)\geq \omega(x_1-x_2)$. Using $\omega\leq b$, we find\begin{equation}
		\begin{aligned}
			&E_1^{(3)}=\binom{N}{2}\int_{B_{12}} \left(\sum_{2\leq i<j\leq N}^{}v_{ij}\abs{\Psi_{12}}^2+\sum_{k=3}^{N}v_{1k}\abs{\Psi_{12}}^2\right)\\&\leq \text{const. }\left(\int_{\{\abs{x_1-x_2}<b\}\cap\{\abs{x_3-x_4}<b\}}v(\abs{x_3-x_4})\abs{\omega(x_3-x_4)}^2\frac{1}{(x_1-x_2)^2}\rho^{(4)}(x_1,x_2,x_3,x_4)\right.\\
			&\qquad\qquad\qquad\left.+\int_{\{\abs{x_1-x_2}<b\}\cap\{\abs{x_2-x_3}<b\}}v(\abs{x_2-x_3})\abs{\omega(x_2-x_3)}^2\frac{1}{(x_1-x_2)^2}\rho^{(3)}(x_1,x_2,x_3)\right.\\&\qquad\qquad\qquad\left.
			+b^2\int_{\{\abs{x_1-x_2}<b\}\cap\{\abs{x_3-x_4}>b\}}v(\abs{x_3-x_4})\frac{1}{(x_1-x_2)^2}\rho^{(4)}(x_1,x_2,x_3,x_4)\right.\\
			&\qquad\qquad\qquad\left.+b^2\int_{\{\abs{x_1-x_2}<b\}\cap\{\abs{x_2-x_3}>b}v(\abs{x_2-x_3})\frac{1}{(x_1-x_2)^2}\rho^{(3)}(x_1,x_2,x_3)\right).
		\end{aligned}
	\end{equation}
	Hence, by Lemmas \ref{lemscatlength} and \ref{LemmaDensityBounds} and noting some of the integrals give factors of $L$ and $b$, we find that
	\begin{equation}
		\begin{aligned}
			E_1^{(3)}&\leq \text{const. } \Bigg[N^2(\rho b)\rho^5\left(\int_{\abs{x
				}<b} x^2 v(x)\abs{\omega(x)}^2\diff x+b^2\int_{\abs{x
				}>b} x^2 v(x)\diff x\right)\\&\qquad\qquad\qquad +N(\rho b) \rho^5 \left(\int_{\abs{x}
				<b} x^2 v(x)\abs{\omega(x)}^2\diff x+b^2\int_{\abs{x}
				>b} x^2 v(x)\diff x\right)\Bigg]\\
			&\quad \leq \text{const. } \left[N^2(\rho b)^5\rho/(b-a_b)+N^2(\rho b)^3\rho ^3 \int_{\abs{x}>b} v(x)x^2\right]\\
			&\quad=\text{const. }E_0 N (\rho b)^3 \left(\underbrace{\rho b^2/(b-a_b)}_{<2}+\rho \int_{x>b} v(x)x^2\right),
		\end{aligned}
	\end{equation}
where we used the assumptions of Lemma \ref{LemmaE1Bound} to bound the fraction by $2$. 
Given that $\Psi_F$ is an eigenfunction of $-\sum_{i=1}^{N}\partial_i^2$ with eigenvalue $E_0$, the error $E_1^{(4)}$ satisfies
\begin{equation}
	\begin{aligned}
		&E_1^{(4)}=\binom{N}{2}\int_{A_{12}}\sum_{i=3}^{N} \overline{\Psi_F}\frac{\omega(x_1-x_2)^2}{(x_1-x_2)^2}(-\partial^2_i)\Psi_F\\&\quad=E_0\binom{N}{2}\int_{A_{12}}\left\lvert\frac{\omega(x_1-x_2)}{x_1-x_2}\Psi_F\right\rvert^2-2\binom{N}{2}\int_{A_{12}} \overline{\Psi_F}\frac{\omega(x_1-x_2)^2}{(x_1-x_2)^2}(-\partial^2_1)\Psi_F.
	\end{aligned}
\end{equation}
 By Lemma \ref{LemmaDensityBounds} and $\abs{\omega}\leq b$, it follows that
	\begin{equation}
		\binom{N}{2}\int_{A_{12}}\left\lvert\frac{\omega(x_1-x_2)}{x_1-x_2}\Psi_F\right\rvert^2\leq \frac{b^2}{2}\int_{\{\abs{x_1-x_2}<b\}} \frac{\rho^{(2)}(x_1,x_2)}{\abs{x_1-x_2}^2}\diff x_1\diff x_2\leq \text{const. }N(\rho b)^3,
	\end{equation}
	and also by Lemma \ref{LemmaDensityBounds}, \begin{equation}
		\begin{aligned}
			2\binom{N}{2}\abs{\int_{A_{12}} \overline{\Psi_F}\frac{\omega(x_1-x_2)^2}{(x_1-x_2)^2}(-\partial^2_1)\Psi_F}&\leq\frac12\abs{\sum_{i=1}^{2}\int_{A_{12}}\frac{\omega(x_1-x_2)^2}{(x_1-x_2)^2}\partial^2_{y_i}\gamma^{(2)}(x_1,x_2;y_1,y_2)\Big\rvert_{y=x}}\\&\leq \text{const. } N\rho^2(\rho b)^3,
		\end{aligned}
	\end{equation}
Hence, $E_1^{(4)}\leq \text{const. }E_0(N(\rho b)^3)$.
Combining all bounds above, we find the result \eqref{eqnlemma12}.
\end{proof}
\subsection{Estimating $E_2^{(1)}+E_2^{(2)}$}
\label{secE2}
Recall that \begin{equation}
	\begin{aligned}
		E_2^{(1)}&=\binom{N}{2}2N\int_{A_{12}\cap A_{13}}\sum_{i=1}^{N}\abs{\partial_i\Psi_F}^2,\\ E_2^{(2)}&=\binom{N}{2}\binom{N-2}{2}\int_{A_{12}\cap A_{34}}\sum_{i=1}^{N}\abs{\partial_i\Psi_F}^2,
	\end{aligned}
\end{equation}
with $A_{ij}:=\{x\in\R^N\vert \abs{x_i-x_j}<b\}$. We prove the following bound. 
\begin{lemma}\label{LemmaE2Bound}
	\begin{equation}
		E_2^{(1)}+E_2^{(2)}\leq \textnormal{const. }E_0\left(N(\rho b)^4+N^2(\rho b)^6\right).
	\end{equation}
\end{lemma}
\begin{proof}
	We start by splitting $ E_2^{(1)} $ and $ E_2^{(2)} $ in two terms each and using partial integration. Consider first $ E_2^{(1)} $,  
	\begin{equation}
		\begin{aligned}
			E_2^{(1)}&=\binom{N}{2}2N\int_{A_{12}\cap A_{23}}\sum_{i=1}^{N}\abs{\partial_i\Psi_F}^2\\
			&=\binom{N}{2}2N\int_{A_{12}\cap A_{23}}\sum_{i=1}^{3}\abs{\partial_i\Psi_F}^2+\binom{N}{2}2N\int_{A_{12}\cap A_{23}}\sum_{i=4}^{N}\abs{\partial_i\Psi_F}^2.
		\end{aligned}
	\end{equation}
	For the second term, we perform partial integration, and use that $\Psi_F$ is an eigenfunction of the Laplacian with eigenvalue $E_0$ to find \begin{equation}
		\begin{aligned}
			\binom{N}{2}&2N\int_{A_{12}\cap A_{23}}\sum_{i=4}^{N}\abs{\partial_i\Psi_F}^2=\binom{N}{2}2N\int_{A_{12}\cap A_{23}}\sum_{i=4}^{N}\overline{\Psi_F}(-\partial^2_i\Psi_F)\\
			&\leq E_0 \binom{N}{2}2N\int_{A_{12}\cap A_{23}}\abs{\Psi_F}^2-\binom{N}{2}2N\int_{A_{12}\cap A_{23}}\sum_{i=1}^{3}\overline{\Psi_F}(-\partial^2_i\Psi_F)\\&\leq 3E_0\int_{[0,L]}\int_{[x_2-b,x_2+b]}\int_{[x_2-b,x_2+b]}\rho^{(3)}(x_1,x_2,x_3)\diff x_3\diff x_1\diff x_2\\
			&\qquad  -\binom{N}{2}2N\int_{A_{12}\cap A_{23}}\sum_{i=1}^{3}\overline{\Psi_F}(-\partial^2_i\Psi_F).
		\end{aligned}
	\end{equation}
	Lemma \ref{LemmaDensityBounds} implies \begin{equation}
		\begin{aligned}
			3E_0\int_{[0,L]}\int_{[x_2-b,x_2+b]}\int_{[x_2-b,x_2+b]}\rho^{(3)}(x_1,x_2,x_3)\diff x_3\diff x_1\diff x_2\leq \text{const. }NE_0(\rho b)^6.
		\end{aligned}
	\end{equation}
Lemma \ref{LemmaDensityBounds} also implies \begin{equation}
		\binom{N}{2}2N\int_{A_{12}\cap A_{23}}\sum_{i=1}^{3}\left(\abs{\partial_i\Psi_F}^2-\overline{\Psi_F}(-\partial^2_i\Psi_F)\right)\leq \text{const. }\rho^9 L b^6\leq\text{const. }E_0 (\rho b)^6,
	\end{equation}
 so that $E_2^{(1)}$ is bounded by $\text{const. }NE_0(\rho b)^6$.
	We estimate $ E_2^{(2)} $ in the same way: integration by parts and antisymmetry give  \begin{equation}
		\begin{aligned}
			E_2^{(2)}&=\binom{N}{2}\binom{N-2}{2}\int_{A_{12}\cap A_{34}} \left(\sum_{i=1}^{4}\abs{\partial_i\Psi_F}^2+\sum_{i=5}^{N}\abs{\partial_i\Psi_F}^2\right)\\
			&=\binom{N}{2}\binom{N-2}{2}\left(4\int_{\abs{x_3-x_4}<b}\left[\overline{\Psi_F}\partial_1\Psi_F\right]_{x_1=x_2-b}^{x_1=x_2+b} +\int_{A_{12}\cap A_{34}} \sum_{i=1}^{N}\overline{\Psi_F}(-\partial^2_i\Psi_F)\right)\\
			&=4\int_{x_2\in[0,L]}\int_{\abs{x_3-x_4}<b}\left[\partial_{y_1}\gamma^{(4)}(x_1,x_2,x_3,x_4;y_1,x_2,x_3,x_4)\bigg\vert_{y_1=x_1}\right]_{x_1=x_2-b}^{x_1=x_2+b}\\&\hspace{1cm}+E_0\int_{A_{12}\cap A_{34}}\rho^{(4)}(x_1,\dots,x_4).
		\end{aligned}
	\end{equation}
	Lemma \ref{LemmaDensityBounds} implies that \begin{equation}
		4\int_{x_2\in[0,L]}\int_{\abs{x_3-x_4}<b}\left[\partial_{y_1}\gamma^{(4)}(x_1,x_2,x_3,x_4;y_1,x_2,x_3,x_4)\bigg\vert_{y_1=x_1}\right]_{x_1=x_2-b}^{x_1=x_2+b}\leq\text{const. }E_0 N (\rho b)^4,
	\end{equation}
	and\begin{equation}
		E_0\int_{A_{12}\cap A_{34}}\rho^{(4)}(x_1,\dots,x_4)\leq \text{const. } E_0 N^2(\rho b)^6,
	\end{equation}
	which give the estimate of $E^{(2)}_2$.
\end{proof}

\subsection{Estimating $E_3$}  \label{secE3}
\begin{lemma} \label{LemmaE3Bound}
	$$E_3\leq \textnormal{const. } N\rho^3 \int_{b}^\infty v(x)x^2$$
\end{lemma}
\begin{proof}
    The bound $\rho^{(2)}(x_1,x_2)\leq 8\pi^2\rho^4(x_1-x_2)^2$ is given in Lemma \ref{LemmaDensityBounds}. It follows that
$$
\begin{aligned}
E_3=\binom{N}{2}\int_{A_{12}^\complement} v_{12}\abs{\Psi_F}^2&=\int_{\abs{x_1-x_2}>b}\rho^{(2)}(x_1,x_2)v(x_1-x_2)\\&\leq 8\pi^2 N\rho^3 \int_{b}^{\infty}v(x) x^2.
\end{aligned}
$$

\end{proof}


\subsection{Proof of Proposition \ref{PropositionUpperBound}: A trial state for arbitrary $N$}
\label{secarbN}

The strategy of the proof of Proposition \ref{PropositionUpperBound} is as follows. 
Together, Lemmas \ref{LemmaEnergyFunctionalBound}, \ref{LemmaE1Bound}, \ref{LemmaE2Bound} and \ref{LemmaE3Bound} provide a proof of Lemma \ref{LemmaUpperBoundFewParticles}, which is the upper bound for small $N$ obtained from the trial state $\Psi_\omega$ \eqref{psiomega}. 
To construct a trial state for arbitrary $N$, we glue together copies of $\Psi_\omega$ on small intervals. This is straightforward with Dirichlet boundary conditions since the wave functions vanish at the boundaries. 

\begin{proof}[Proof of Proposition \ref{PropositionUpperBound}]
 Let $\rho\abs{a}$ be sufficiently small to fix $b$ to satisfy $$\frac{(\rho\abs{a})^{4/5}}{\rho}<b\leq \frac{1}{16}\frac1{\rho}.$$ 
 Then $(\rho\abs{a})^{4/5}<\frac{1}{16}$ which implies $\rho b>(\rho \abs{a})^{-1/5}\rho \abs{a}\geq 2\rho \abs{a}$. Therefore we find
	\begin{equation}
 \label{someeq8}
		\frac{b}{b-a}\leq1+2a/b\leq  1+2(\rho\abs{a})^{1/5}.
	\end{equation}
 Notice first that one `box' suffices for $N<\frac{1}{2}(\rho b)^{-3/2}$: the result follows directly from Lemmas \ref{LemmaUpperBoundFewParticles} and $\ref{LemmaDensityBounds}$ with the observation that
  \begin{equation}
  \label{someeq9}
  \norm{\Psi_\omega}^2\geq 1-\int_B\abs{\Psi_F}^2\geq 1-\int_{\abs{x_1-x_2}<b}\rho^{(2)}(x_1,x_2)\geq 1-\frac{16\pi^2}{3}N(\rho b)^3,  
\end{equation}
as well as $\frac{16\pi^2}{3}N(\rho b)^3\leq \frac{8\pi^2}{3}(\rho b)^{3/2}< 1/2$, so that $\frac{1}{1-N(\rho b)^3}\leq 1+2 N(\rho b)^3$.

For $N\geq \frac{1}{2}(\rho b)^{-3/2}$, we consider $M$ disjoint `boxes' (intervals) $I_1,I_2,...,I_M$, with $\bigcup_{k=1}^{M}I_M\subset [0,L]$. We then set the number of particles in each box to be either $\ceil{N/M}$ or $\ceil{N/M}-1$, adding up to $N$ in total. We choose $$
M=\ceil{2N(\rho b)^{3/2}},
$$
 so that
  $$ \frac14(\rho b)^{-3/2}\leq \ceil{N/M}-1 <\ceil{N/M}\leq \frac{1}{2}(\rho b)^{-3/2}+1\leq \frac{33}{64}(\rho b)^{-3/2}.$$
  Hence the number of particles in each box will be between $\frac{1}{4}(\rho b)^{-3/2}$ and $\frac{33}{64}(\rho b)^{-3/2}$.
For convenience, we denote $\tilde{N}=\ceil{N/M}$, and fix the length of each box so that the density is $\tilde{\rho}=\frac{N}{L-(M-1)b}$. In that way the combined length of all $M$ boxes add up to $L-(M-1)b$, leaving a distance $ b $ between each box. Also,  \begin{equation}\label{EqTilderhoEstimate}
    \tilde{\rho}=\frac{\rho}{1-\frac{b(M-1)}{L}}\leq \rho(1+2b(M-1)/L)\leq \rho(1+4(\rho b)^{5/2}), 
\end{equation} which follows from $ b(M-1)/L\leq 2(\rho b)^{5/2}< 1/2 $. 

We now consider the glued trial state $\Psi_{\text{full}}=\prod_{i=1}^{M}\Psi_{\omega,k}$, where $\Psi_{\omega,k}$ denotes the state defined in \eqref{psiomega}, on the box $I_k$.
Besides the energy in each box, we should estimate the long-range interaction energy between boxes. Denoting the coordinates in box $I_k$ by $x^k_1,...x^k_{\tilde{N}}$ (or $x^k_1,...x^k_{\tilde{N}-1}$ if box $I_k$ has $\tilde{N}-1$ particles), we have the interaction energy between boxes given by\begin{equation}
	\begin{aligned}
		\frac{\sum_{k<l}\sum_{i<j}\int \overline{\Psi_{\textnormal{full}}(x)} \Psi_{\textnormal{full}}(x)v(x^k_i-x^l_j)}{\norm{\Psi_{\textnormal{full}}}^2},
	\end{aligned}
\end{equation}
 where $k<l$ sum over boxes and $i<j$ over the particles in each box. Note that the first double sum runs over integrals over disjoint regions of $[0,L]^2$. Hence, using $\tilde{\rho}\leq \rho (1+4(\rho b)^{5/2})\leq 2\rho$, we find that 
\begin{equation}\label{EqEnergyBetweenBoxes}
	\begin{aligned}
		&\frac{\sum_{k<l}\sum_{i<j}\int \overline{\Psi_{\textnormal{full}}(x)} \Psi_{\textnormal{full}}(x)v(x^k_i-x^l_j)}{\norm{\Psi_{\textnormal{full}}}^2}=\sum_{k\neq l}\frac{\int_{y-x>b} \rho^{(1)}_{\Psi_{\omega,k}}(x)\rho^{(1)}_{\Psi_{\omega,l}}(y)v(x-y)}{2\norm{\Psi_{\omega,k}}^2\norm{\Psi_{\omega,l}}^2}\\
        &\leq \frac{2\sum_{k=1}^{M}\int_{y-x>b}\mathbbm{1}_{I_k}(x) \tilde{\rho}^2(1+\textnormal{const. }(\tilde{\rho} b)^3)^2v(x-y)}{\min_{k}\norm{\Psi_{\omega,k}}^4} \\
		&\leq \frac{\textnormal{const. }}{\min_k\norm{\Psi_{\omega,k}}^4}N \rho\int_{x>b} v(x),
	\end{aligned}
\end{equation}
where we used that $\rho^{(1)}_{\Psi_{\omega,k}}$ is supported in $I_k$ for any $k$, as well as the bound \begin{equation}
\begin{aligned}
    	\rho^{(1)}_{\Psi_\omega}(x)\leq \rho^{(1)}_{\Psi_F}(x)+b^2\int_{\{x_2| \abs{x-x_2}<b\}}\frac{1}{(x-x_2)^2}\rho^{(2)}(x,x_2)\diff x_2\\
     \leq 2\tilde{\rho}(1+\textnormal{const. }(\tilde{\rho} b)^3),
\end{aligned}
\end{equation}
which follows from Lemma \ref{LemmaDensityBounds} and \eqref{derivsbound}.

 To study the energy within each box, let $ e_{0,I_k}=\frac{\pi^2}{3}\abs{I_k}\tilde{\rho}^3(1+\text{const. }\tilde{N}^{-1}) $ be the corresponding free Fermi energy in box $I_k$ (where $\abs{I_k}$ is the length of box $I_k$).
 From Lemma \ref{LemmaUpperBoundFewParticles} and \eqref{EqEnergyBetweenBoxes}, the trial state constructed using the small boxes gives the upper bound \begin{equation}\label{EqUpperBoundSmallN}
		\begin{aligned}
			E^D(N,L)\leq \frac{1}{\norm{\Psi_\omega}^2}\sum_{k=1}^{M}e_{0,I_k}\Bigg[ 1+2\tilde{\rho} a\frac{b}{b-a} &+ \text{const. } \tilde{N} (\tilde{\rho}b)^3\\
   &+\textnormal{const. }\left(\tilde{\rho} \int_{x>b} v(x)x^2+\tilde{\rho} \abs{a} \frac{\ln(\tilde{N})}{\tilde{N}}\right)\Bigg]\\
   &+\textnormal{const. } \frac{N\rho}{\min_k\norm{\Psi_{\omega,k}}^4} \int_{x>b} v(x).
		\end{aligned}
	\end{equation}
 Notice that $\tilde{\rho} \abs{a} \frac{\ln(\tilde{N})}{\tilde{N}}\leq  \frac{2}{e}\max (\tilde{N}^{-1},(\tilde{\rho} \abs{a})^{3/2} )$ (this  is easily seen by considering the cases $\tilde{N}\leq(\tilde\rho \abs{a})^{-1}$ and $\tilde{N}>(\tilde\rho \abs{a})^{-1}$ separately). Thus this term is subleading, and we absorb it into other error terms.
Using \eqref{someeq8}, \eqref{someeq9}, and \eqref{EqTilderhoEstimate} as well as the observations above, we see that \begin{equation}
		\begin{aligned}
			E\leq N\frac{\pi^2}{3}\rho^2\frac{\left(1+2\rho a+4(\rho \abs{a})^{6/5}+\text{const. }\left[\frac{M}{N}+\tilde{N}(\rho b)^3+\rho \int_{x>b} v(x) x^{2}\right]\right)}{1-\frac{16\pi^2}3 \tilde{N}(\rho b)^3}\\
			+\textnormal{const. } \frac{N\rho}{(1-\frac{16\pi^2}{3}\tilde{N}(\rho b)^3)^2}\int_{x>b} v(x).
		\end{aligned}
	\end{equation}
	Using $M/N\leq (\rho b)^{3/2}+\frac{1}{N}$ and the fact that $ \frac{16\pi^2}{3}\tilde{N}(\rho b)^3\leq 1/2 $, so that $$ 1/(1-\tilde{N}(\rho b)^3)\leq 1+2\tilde{N}(\rho b)^3, $$
 as well as $\tilde{N}(\rho b)^3\leq \textnormal{const. }(\rho b)^{3/2}$,
 the desired result follows by minimizing over $b$ in the chosen range. Note the error $N\frac{\pi^2}{3}\rho^2\mathcal{O}(N^{-1})$ in the desired result is only important when $N<(\rho b)^{-3/2}$, so that it does not appear in this case.

\end{proof}

\section{Lower bound in Theorem \ref{TheoremMain}}	
\label{SecLowerbound}
\begin{proposition}[Lower bound in Theorem \ref{TheoremMain}]
	\label{PropositionLowerBound}
	Consider a Bose gas with repulsive interaction, $v$, supported in $[-R_0,R_0]$ with Neumann boundary conditions. Write $\rho=N/L$. There exists a constant $C_\text{L}>0$ such that the ground state energy $E^N(N,L)$ satisfies
	\begin{equation}
		\label{eqlower}
		E^N(N,L)\geq N\frac{\pi^2}{3}\rho^2\left(1+2\rho a-C_\text{L}\left((\rho\abs{a})^{6/5}+(\rho R_0)^{6/5})+N^{-2/3}\right)\right).
	\end{equation}
\end{proposition}

The result for potentials with non-compact support, as stated in the main result Theorem \ref{TheoremMain}, follows by noting that in a lower bound, we may replace a non-compactly supported potential $v$ by $\mathbbm{1}_{[-R_0,R_0]}v$ for some fixed $R_0>0$. The resulting lower bound contains the cut-off scattering length $a_{R_0}$, but \eqref{estcutoff} allows us to replace $a_{R_0}$ by the full scattering length $a$ at the cost of an error $2\int_{b>R_0}v(x)x^2\diff x$, which, after an optimization in $R_0$, goes into $\varepsilon_v$ in Theorem \ref{TheoremMain}.

As mentioned in Section \ref{SecProofidea}, the proof is based on a reduction to the Lieb-Liniger model combined with Lemma \ref{lemscatlength}. Similar to the upper bound, this idea only provides a useful lower bound for small $N$, which we obtain in Proposition \ref{PropositionLowerBoundSpecN} and Corollary \ref{CorollaryLowerBoundSpecN} at the end Section \ref{seclowsmalln}, after preparatory estimates on the Lieb--Liniger model in Section \ref{secprep}. Then, in Section \ref{seclowboundarbn}, this lower bound will be generalized to arbitrary $N$, proving Proposition \ref{PropositionLowerBound}.

\subsection{Lieb-Liniger model: preparatory facts}
\label{secprep}
The ground state energy $E_{LL}(N,L,c)$ of the Lieb-Liniger model with $N$ particles in an interval of length $L$ is determined by the system of equations \cite{lieb1963exact}
\begin{align}
	k_{i+1}-k_i&=\frac{2\pi}{L}+\frac{2}{L}\sum_{s=1}^{N}\left(\arctan\left(\frac{k_s-k_{i+1}}{c}\right)-\arctan\left(\frac{k_s-k_{i}}{c}\right)\right),\label{EqLLGroundState1}\\
	\sum_{i=1}^{N}k_i&=0,\label{EqLLGroundState2}\\
	E_{LL}(N,L,c)&=\sum_{i=1}^{N}k_i^2.\label{EqLLGroundState3}
\end{align}
It was shown by Yang and Yang in \cite{yang1969thermodynamics} that these equations have a unique solution for the pseudo-momenta $k_i$ and hence the energy $E_{LL}(N,L,c)$. It is not difficult to show that this solution gives the ground state. As an stronger result, Dorlas proved in \cite{cmp/1104252974} that the solutions found in \cite{lieb1963exact} form a complete orthogonal basis. 

Notice that \eqref{EqLLGroundState1} implies that $k_{i+1}-k_i\leq \frac{2\pi}{L}$ (with $k_{i+1}\geq k_i$ by construction). By Taylor expanding the right-hand side of \eqref{EqLLGroundState1} in $k_{i+1}$ around $k_i$ to first order, the pseudo-momenta of the ground state satisfy
\begin{equation}\label{EqPseudoMomentaDifferenceExpansion}
	k_{i+1}-k_i=\frac{2\pi}{L}-\frac{1}{L}\sum_{s=1}^{N}\frac{2c(k_{i+1}-k_i)}{c^2+(k_s-k_i)^2}+\rho O(1/(cL)^2).
\end{equation}
To give a lower bound on the ground state energy, we use the following lemma.
\begin{lemma}\label{LemmaLL-LowerBoundFiniteN}
	Let $ (k_i)_{i=1}^{N} $ satisfy $ k_1<k_2<\ldots<k_N $, \eqref{EqLLGroundState1}, and \eqref{EqLLGroundState2}. Then we have 
	\begin{equation}
 \label{eqlemma3}
		E_{LL}(N,L,c)=\sum_{i=1}^{N}k_i^2\geq N\left(\rho^2-\frac{1}{L^2}\right)\frac{\pi^2}{3}\left(1+2\frac{\rho}{c}\right)^{-2}+\mathcal{O}\left(\frac{\rho^4}{c^2} \right).
	\end{equation}
\end{lemma}
\begin{proof} 
	As noted above, the pseudo-momenta satisfy \eqref{EqPseudoMomentaDifferenceExpansion}.
	By discarding the term $ (k_s-k_i)^2 $ in the denominator inside the sum, we find
	$$k_{i+1}-k_i\geq \frac{2\pi}{L}\left(1+2\frac{\rho}{c}\right)^{-1}+\rho \mathcal{O}\left(1/(cL)^2\right). $$
	For the ground state, we have $k_i=-k_{N-i+1}$ (this follows from reflection symmetry of the equations and the uniqueness of the ground state; see Section III of \cite{lieb1963exact}). Hence,  for all $1\leq i\leq N$,
 \begin{equation}
		2\abs{k_i}\geq \frac{2\pi}{L}\left(1+2\frac{\rho}{c}\right)^{-1}|N+1-2i|+N\rho \mathcal{O}\left(1/(cL)^2\right).
	\end{equation} 
	Therefore, we find the lower bound \eqref{eqlemma3} by using \eqref{EqLLGroundState3} and estimating the sum. 
\end{proof}

The thermodynamic Lieb--Liniger energy density behaves like $\rho^3 e(c/\rho)$. By letting $N,L\to\infty$ with $N/L\to\rho$, the  previous lemma gives lower bound on this energy density.
\begin{lemma}[Lieb-Liniger lower bound] \label{LemmaLL-LowerBound}
	For $\gamma>0$,
	\begin{equation}
		\label{somresult}
		e(\gamma)\geq \frac{\pi^2}{3}\left(\frac{\gamma}{\gamma+2}\right)^2\geq \frac{\pi^2}{3}\left(1-\frac{4}{\gamma}\right).
	\end{equation}
\end{lemma}
The next result corrects the lower bound from \eqref{somresult} to obtain an estimate for finite particle numbers $n$ with Neumann boundary conditions.
\begin{lemma}[Lieb-Liniger lower bound for finite $n$]\label{LemmaLiebLinigerNeumannLowerBound}
	The Lieb--Liniger ground state energy with Neumann boundary conditions can be estimated by
	\begin{equation}
		E_{LL}^{N}(n,\ell,c)\geq \frac{\pi^2}{3}n\rho^2\left(1-4\rho/c-\textnormal{const. }\frac{1}{n^{2/3}}\right).
	\end{equation}
\end{lemma}
This will be proved after the following lemma. Note we use the superscripts $N$ and $D$ to denote Neumann and Dirichlet boundary conditions, respectively. For simplicity, we will consider the Lieb-Liniger model on $[-L/2,L/2]$ in this subsection, and use the notation $\Lambda_s:=[-s/2,s/2]$.
\begin{lemma}[Robinson \cite{robinson2014thermodynamic}]\label{LemmaRobinson}
	Let $ v$ be symmetric and decreasing (that is, $ v\circ \mathfrak{c}\geq v $ for any contraction $ \mathfrak{c} $). For any $ b>0 $,  \begin{equation}\label{EqRobinsonBound}
		E^D_{\Lambda_{L+2b}}\leq E^N_{\Lambda_L}+\frac{2n}{b^2}.
	\end{equation}
\end{lemma}
\begin{proof}
	The idea of the proof is given on page 66 of \cite{robinson2014thermodynamic}, but we shall give a more explicit proof here. In order to compare energies with different boundary conditions, consider a cut-off function $ h $ with the property that
	\begin{enumerate}
		\item $ h $ is real, symmetric, and continuously differentiable on $ \Lambda_{3L} $,
		\item $ h(x)=0 $ for $ \abs{x}>L/2+b $,
		\item $ h(x)=1 $ for $ \abs{x}<L/2-b $,
		\item $ h(L/2-x)^2+h(L/2+x)^2=1 $ for $ 0<x<b $,
		\item $ \abs{\frac{\diff h}{\diff x}}^2\leq \frac{1}{b^2} $, and $ h^2\leq 1 $.
	\end{enumerate}

Let $ f\in \mathcal{D}(\mathcal{E}^N_{\Lambda_L}) $. Define $ \tilde{f} $ by extending $ f $ to $ \Lambda_{3L} $ by reflecting $ f $ across each face of its domain in $ \Lambda_{3L} $. Define then $ V:L^2(\Lambda_L)\to L^2(\Lambda_{L+2b})  $ by $ Vf(x):=\tilde{f}(x)\prod_{i=1}^{n}h(x_i) $. It is not hard to show that $ V $ is an isometry, this is shown in Lemma 2.1.12 of \cite{robinson2014thermodynamic}. Also, we clearly have $ Vf\in \mathcal{D}(\mathcal{E}^D_{\Lambda_{L+2b}})  $.  Let $ \psi $ be the ground state for $ \mathcal{E}^N_{\Lambda_L} $, and define the trial state $ \psi_{\text{trial}}=V\psi $. For the
non-interacting system, the bound \eqref{EqRobinsonBound} is obtained in Lemma 2.1.13 of \cite{robinson2014thermodynamic}. It remains to show that
the trial state does not have a higher potential energy than the original state. To verify this, define $ \tilde{\psi} $ to be $ \psi $ extended by reflection as above. For $ \abs{x_2}<L/2-b $,  it suffices that\begin{equation}
		\begin{aligned}
			&\int_{-L/2-b}^{L/2+b}v(\abs{x_1-x_2})\abs{\tilde{\psi}(x)}^2h(x_1)^2h(x_2)^2\diff x_1\\&\leq\int_{-L/2+b}^{L/2-b}v(\abs{x_1-x_2})\abs{\tilde{\psi}(x)}^2\diff x_1+\sum_{s\in\{-1,1\}}s\int_{s(L/2-b)}^{s(L/2)}v(\abs{x_1-x_2})\abs{\tilde{\psi}(x)}^2(h(x_1)^2+h(L-x_1)^2)\diff x_1\\
			&\quad =\int_{-L/2}^{L/2}v(\abs{x_1-x_2})\abs{\tilde{\psi}(x)}^2\diff x_1,
		\end{aligned}
	\end{equation}
where we used that $ v $ is symmetric decreasing, as well as property 4 of $h$. For \hbox{$\abs{x_2}\geq L/2-b$}, writing $\bar{x}^{1,2}$ as shorthand for $ (x_3,\dots, x_N)$, we find
	\begin{equation}
		\begin{aligned}
			&\int_{L/2-b}^{L/2+b}\int_{L/2-b}^{L/2+b}v(\abs{x_1-x_2})\abs{\tilde{\psi}(x)}^2h(x_1)^2h(x_2)^2\diff x_2\diff x_1\\
			&\quad\quad=\sum_{(s_1,s_2)\in\{-1,1\}^2}s_1s_2\int_{L/2-s_1b}^{L/2}\int_{L/2-s_2b}^{L/2}v(\abs{x_1-x_2})\abs{\tilde{\psi}(x)}^2h(x_1)^2h(x_2)^2\diff x_2\diff x_1\\
			&\quad\quad =\sum_{(s_1,s_2)\in\{-1,1\}^2}\int_{0}^{b}\int_{0}^{b}v(\abs{s_1y_1-s_2y_2})\abs{\tilde{\psi}(L/2-s_1 y_1,L/2-s_2 y_2,\bar{x}^{1,2})}^2\\&\hspace{5cm}\times h(L/2-s_1 y_1)^2h(L/2-s_2 y_2)^2\diff y_2\diff y_1\\
			&\quad\quad\leq \int_{0}^{b}\int_{0}^{b}v(\abs{y_1-y_2})\abs{\tilde{\psi}(L/2-y_1,L/2- y_2,\bar{x}^{1,2})}^2\\&\hspace{5cm}\times\sum_{(s_1,s_2)\in\{-1,1\}^2}h(L/2-s_1 y_1)^2h(L/2-s_2 y_2)^2\diff y_2\diff y_1\\
			&\quad\quad=\int_{0}^{b}\int_{0}^{b}v(\abs{y_1-y_2})\abs{\tilde{\psi}(L/2-y_1,L/2- y_2,\bar{x}^{1,2})}^2\diff y_2\diff y_1.
		\end{aligned}
	\end{equation}
	In the third line, we use the definition of $ \tilde{\psi} $ as well as the fact that $ \abs{s_1y_1-s_2y_2}\geq \abs{y_1-y_2} $ for $ y_1,y_2\geq 0 $, and in last line, we use property 4 of $ h $.
 By a similar computation, we find 
 \begin{equation}
		\begin{aligned}
			&\int_{ L/2 - b}^{ L/2+b}\int_{- L/2-b}^{- L/2+b}v(\abs{x_1-x_2})\abs{\tilde{\psi}(x)}^2h(x_1)^2h(x_2)^2\diff x_2\diff x_1\\
   &\qquad\qquad \leq \int_{0}^{b}\int_{0}^{b}v(\abs{L-y_1-y_2})\abs{\tilde{\psi}( L/2- y_1,- L/2 + y_2,\bar{x}^{1,2})}^2\diff y_2\diff y_1.
		\end{aligned}
	\end{equation}
	By combining the three bounds above, we clearly have 
	\begin{equation}
		\begin{aligned}
			&\int_{-L/2-b}^{L/2+b}\int_{-L/2-b}^{L/2+b}v(\abs{x_1-x_2})\abs{\tilde{\psi}(x)}^2h(x_1)^2h(x_2)^2\diff x_1\diff x_2\\&\qquad\qquad\qquad\qquad \leq \int_{-L/2}^{L/2}\int_{-L/2}^{L/2}v(\abs{x_1-x_2})\abs{\tilde{\psi}(x)}^2\diff x_1\diff x_2.
		\end{aligned}
	\end{equation}
	The result now follows from the fact that $ V $ is an isometry.
\end{proof}
\begin{proof}[Proof of Lemma \ref{LemmaLiebLinigerNeumannLowerBound}]
	Lemma \ref{LemmaRobinson} implies that for any $ b>0 $ \begin{equation}
		E_{LL}^{N}(n,\ell,c)\geq E_{LL}^D(n,\ell+b,c)-\text{const. }\frac{n}{b^2}\geq E_{LL}(n,\ell+b,c)-\text{const. }\frac{n}{b^2}, 
	\end{equation}
where we used that Dirichlet boundary conditions produce the highest  ground state energy in the second inequality.
	Using Lemma \ref{LemmaLL-LowerBoundFiniteN}, this implies
	\begin{equation}
		E_{LL}^{N}(n,\ell,c)\geq n\left(\rho^2-\frac{1}{\ell^2}\right)\left(\frac{1}{1+b/\ell}\right)^2\frac{\pi^2}{3}\left(1+2\frac{\rho}{c(1+b/\ell)}\right)^{-2}+\mathcal{O}\left(\frac{\rho^4}{c^2} \right)-\text{const. }\frac{n}{b^2}
	\end{equation}
	Optimizing in $ b $, we find \begin{equation}
		E_{LL}^{N}(n,\ell,c)\geq \frac{\pi^2}{3}n\rho^2\left(1-4\rho/c-\text{const. }\frac{1}{n^{2/3}}\right).
	\end{equation}
\end{proof}
\subsection{Lower bound for small particle numbers $n$}
\label{seclowsmalln}
In this subsection, we work our way towards Proposition \ref{PropositionLowerBoundSpecN} and Corollary \ref{CorollaryLowerBoundSpecN}, which provide lower bounds on the Neumann ground state energy. The proof strategy is that of Section \ref{SecProofidea}.
We start by removing the relevant regions of the wave function. Throughout this section, let $ \Psi $ be the Neumann ground state of $\mathcal{E}$ and let $R>\max\left(R_0,2\abs{a}\right)$ be a length, to be fixed later. (Recall that to prove the lower bound in Proposition \ref{PropositionLowerBound}, we assume the potential is supported in $[-R_0,R_0]$, also see the remark below the theorem). Define the continuous function $ \psi\in L^2([0,\ell-(n-1)R]^n) $ by
\begin{equation}
	\label{defpsi}
	\psi(x_1,x_2,\dots,x_n):=\Psi(x_1,R+x_2,\dots,(n-1)R+x_n)
\end{equation}
for $x_1\leq\dots\leq x_n\leq \ell-(n-1)R$, extended symmetrically to other orderings of the particles. 
Our first goal is to prove that almost no weight is lost in going from $\Psi$ to $\psi$, so that the heuristic calculation \eqref{heurist} has a chance of success. The following lemma will be useful.
\begin{lemma}
	For any function $ \phi\in H^1(\R) $ such that $ \phi(0)=0 $, \begin{equation}\label{EqSobolevIneq}
		\int_{[0,R]}\abs{\partial\phi}^2\geq \max_{[0,R]}\abs{\phi}^2/R.
	\end{equation}
\end{lemma}
\begin{proof}
	Write $ \phi(x)=\int_{0}^{x}\phi'(t)\diff t $, and find that \begin{equation}
		\abs{\phi(x)}\leq \int_{0}^{x}\abs{\phi'(t)}\diff t.
	\end{equation}
	Hence $ \max_{x\in[0,R]}\abs{\phi(x)}\leq \int_{0}^{R}\abs{\phi'(t)}\diff t\leq \sqrt{R}\left(\int\abs{\phi'(t)}^2\diff t\right)^{1/2}. $
\end{proof}
We can estimate the norm loss. As $\Psi$ is normalized, we find 
\begin{equation}\label{EqNormBoundBij}
	\begin{aligned}
		\braket{\psi|\psi}=1-\int_{\{x\in\R^n\vert \min_{i,j}\abs{x_i-x_j}<R \}}\abs{\Psi}^2\geq 1-\sum_{i<j}\int_{D^j_i}\abs{\Psi}^2,
	\end{aligned}
\end{equation}
where $ D^j_i:=\{x\in\R^n \vert \mathfrak{r}_i(x)=\abs{x_i-x_j}<R \}$ with $ \mathfrak{r}_i(x):=\min_{j\neq i}(\abs{x_i-x_j}) $. Note $ D^j_i $ is not symmetric in $ i$ and $j $, and that for $j\neq j'$, $ D^j_i\cap D^{j'}_i=\emptyset$ up to sets of Lebesgue measure zero. Also note that $\cup_{i<j}D^j_i=\{x\in\R^n\vert \min_{i,j}\abs{x_i-x_j}<R \}$.  To give a good bound on the right-hand side of \eqref{EqNormBoundBij}, we need the following lemma, upper bounding the norm loss. 
\begin{lemma}\label{LemmaNormLoss}
	For $ \psi $ defined in \eqref{defpsi}, we have \begin{equation}
		\label{eqlemmanormloss}
		1-\braket{\psi|\psi}\leq 8 R^2\sum_{i<j}\int_{D^j_i}\abs{\partial_i \Psi}^2+R(R-a)\sum_{i<j}\int v_{ij} \abs{\Psi}^2.
	\end{equation}
\end{lemma}
\begin{proof}
	Note that \eqref{EqSobolevIneq} implies that for any $ \phi\in H^1 $ and $ x,x'\in[0,R] $, \begin{equation}
		\abs{\abs{\phi(x)}-\abs{\phi(x')}}^2\leq\abs{\phi(x)-\phi(x')}^2\leq R\left(\int_{[0,R]}\abs{\partial \phi}^2\right).
	\end{equation}
	 Furthermore, 
	\begin{equation}
		\begin{aligned}
			\abs{\phi(x)}^2-\abs{\phi(x')}^2=\left(\abs{\phi(x)}-\abs{\phi(x')}\right)^2+2\left(\abs{\phi(x)}-\abs{\phi(x')}\right)\abs{\phi(x')}\\\leq 2\left(\abs{\phi(x)}-\abs{\phi(x')}\right)^2+\abs{\phi(x')}^2.
		\end{aligned}
	\end{equation}
	It follows that \begin{equation}
 \label{someeq12}
		\max_{x\in[0,R]}\abs{\phi(x)}^2\leq 2R\int_{[0,R]}\abs{\partial \phi}^2+2\min_{x'\in[0,R]}\abs{\phi(x')}^2.
	\end{equation}
	Viewing $ \Psi $ as a function of $ x_i $ with the other positions fixed and fixed $j$, \eqref{someeq12} implies \begin{equation}
		2\min_{\mathfrak{r}_i(x)=\abs{x_i-x_j}<R}\abs{\Psi}^2\geq \max_{\mathfrak{r}_i(x)=\abs{x_i-x_j}<R}\abs{\Psi}^2-4R\left(\int_{{\mathfrak{r}_i(x)=\abs{x_i-x_j}<R}}\abs{\partial_i \Psi}^2\right),
	\end{equation}
 where we used that the interval in question has length at most $2R$.
	Hence, \begin{equation}
		\begin{aligned}
			&2\sum_{i<j}\int v_{ij} \abs{\Psi}^2\geq 2\sum_{i<j} \int_{D^j_i} v_{ij} \abs{\Psi}^2 \\&\geq \left(\int v\right)\sum_{i< j}\int\left(\max_{\tilde{D}^j_i}\abs{\Psi}^2-4R\left(\int_{\tilde{D}^j_i}\abs{\partial_i\Psi}^2\diff x_i\right)\right)\diff \bar{x}^i\\
			&\geq \frac{4}{R-a}\sum_{i<j}\left(\frac{1}{2R}\int_{D^j_i}\abs{\Psi}^2-4R\int_{D^j_i}\abs{\partial_i\Psi}^2\right),
		\end{aligned}
	\end{equation}
	where $ \tilde{D}^j_i:=\{x_i\in \R \vert \mathfrak{r}_i(x)=\abs{x_i-x_j}<R \} $ and $\diff \bar{x}^i$ is shorthand for integration with respect to all variables except $x_i$. In the last line, we used $\int v\geq 4/(R-a)$ (see Lemma \ref{lemscatlength}) and lower bounded the maximum of $|\Psi|^2$ by an average. Note that positivity of $v$ is important in this step. Rewriting and \eqref{EqNormBoundBij} give the result.
\end{proof}

To make \eqref{heurist} in the proof outlined in Section \ref{SecProofidea} precise, we relate the Neumann ground state energy to the Lieb--Liniger energy in Lemma \ref{LemmaNormBoundEpsilon}. First, we state a direct adaptation of Lemma \ref{lemscatlength}, more suited to our purpose here (note the boundary conditions $f(R)=f(-R)=1$ in Lemma \ref{lemscatlength} are replaced by the contributions from the delta functions here). 

\begin{lemma}[Dyson's lemma]\label{LemmaDyson} Let $ R>R_0=\textnormal{range}(v) $ and $ \varphi\in H^1(\R) $, then for any interval $ \mathcal{I}\ni 0 $ 
	\begin{equation}
		\int_{\mathcal{I}} \abs{\partial \varphi}^2+\frac12 v\abs{\varphi}^2\geq \int_{\mathcal{I}}\frac{1}{R-a}\left(\delta_R+\delta_{-R}\right)\abs{\varphi}^2,
	\end{equation}
	where $ a $ is the scattering length.
\end{lemma}

For the remaining part of this section we set \begin{equation}
    R=\max\left(R_0,2\abs{a}\right) 
\end{equation}

\begin{lemma}\label{LemmaNormBoundEpsilon}
	Let  $ \epsilon\in[0,1] $. For $ \psi $ defined in \eqref{defpsi},
	\begin{equation}
 \label{someeq1309}
		\int \sum_{i}\abs{\partial_i\Psi}^2+\sum_{i\neq j} \frac{1}{2}v_{ij}\abs{\Psi}^2\geq E_{LL}^N \left(n,\tilde{\ell},\frac{2\epsilon}{R-a}\right)\braket{\psi|\psi}+ \frac{(1-\epsilon)}{8R^2}(1-\braket{\psi|\psi}).
	\end{equation}
	where $ \tilde{\ell}:=\ell-(n-1)R $.
\end{lemma}
\begin{proof}
	We first split the left-hand side of \eqref{someeq1309},  
 	\begin{equation}
		\begin{aligned}
&\int\sum_{i}\abs{\partial_i\Psi}^2\mathds{1}_{\mathfrak{r}_i(x)>R}+\epsilon\left(\sum_{i}\abs{\partial_i\Psi}^2\mathds{1}_{\mathfrak{r}_i(x)<R}+\int\sum_{i<j} v_{ij} \abs{\Psi}^2\right)\\&\qquad\qquad\qquad+ (1-\epsilon)\left(\sum_{i}\abs{\partial_i\Psi}^2\mathds{1}_{\mathfrak{r}_i(x)<R}+\int\sum_{i<j} v_{ij} \abs{\Psi}^2\right).
		\end{aligned}
	\end{equation}
 Using Lemma \ref{LemmaDyson} on the second term (see also \eqref{eqidea}), and using the fact that $\{\mathfrak{r}_i(x)<R\}$ is the disjoint union $\bigcup_{j}D_i^j$, we find that the left-hand side of \eqref{someeq1309} is lower bounded by
	\begin{equation}\label{someeq13}
		\begin{aligned}
&\int\sum_{i}\abs{\partial_i\Psi}^2\mathds{1}_{\mathfrak{r}_i(x)>R}+\epsilon\sum_{i}\frac{1}{R-a}\delta(\mathfrak{r}_i(x)-R)\abs{\Psi}^2\\&\qquad\qquad\qquad+ (1-\epsilon)\left(\sum_{i<j}\int_{D^j_i}\abs{\partial_i \Psi}^2+\int\sum_{i<j} v_{ij} \abs{\Psi}^2\right),
		\end{aligned}
	\end{equation}
	where $ \mathfrak{r}_i(x)=\min_{j\neq i}(\abs{x_i-x_j}) $ and the nearest neighbour delta interaction can be written $\delta(\mathfrak{r}_i(x)-R)=\left(\sum_{j\neq i}\left[\delta(x_i-x_j-R)+\delta(x_i-x_j+R)\right]\right)\mathbbm{1}_{\mathfrak{r}_i(x)\geq R}$. The nearest-neighbour interaction is obtained by, for each $i$ in the sum above, dividing the integration domain of $x_i$ into Voronoi cells around $x_k$ with $k\neq i$. Then, for each $k$, restricting to the cell around particle $ k $ and using Lemma \ref{LemmaDyson} gives the desired nearest neighbour interaction. This technique is also used in \cite{lieb2006mathematics}.
 
	With the use of Lemma \ref{LemmaNormLoss} with $ R>2\abs{a} $ in the last line of \eqref{someeq13}, and by realizing that the first two terms can be obtained by using $ \psi $ as a trial state in the Lieb-Liniger model (since the two delta functions collapse to a single delta of twice the strength when volume $R$ is removed between particles), we obtain the claimed bound \eqref{someeq1309}. 
\end{proof}

The next lemma bounds the difference in norm between $ \Psi $ of norm $ 1 $ and $ \psi $ in \eqref{defpsi}. 
\begin{lemma}\label{LemmaImprovedMassBound}
	For $ n(\rho R)^2\leq  \frac{3}{16\pi^2}\frac{1}{8} $ and $ \rho R\leq \frac{1}{2} $ we have
	\begin{equation}\label{EqImprovedMassBound}
		\begin{aligned}
			\braket{\psi|\psi} \geq 1-\textnormal{const. }\left(n(\rho R)^3+n^{1/3}(\rho R)^2\right).
		\end{aligned}
	\end{equation}
\end{lemma}
\begin{proof}
	Our upper bound (Proposition \ref{PropositionUpperBound} applied to a potential with range $R_0$) and Lemma \ref{LemmaNormBoundEpsilon} with $ \epsilon=1/2$ imply  that 
	\begin{equation}
		n\frac{\pi^2}{3}\rho^2\left(1+2\rho a+\text{const. }(\rho R)^{6/5}\right)\geq E_{LL}^N \left(n,\tilde{\ell},\frac{1}{R-a}\right)\braket{\psi|\psi}+ \frac{1}{16R^2}(1-\braket{\psi|\psi}).
	\end{equation}
	Subtracting $ E_{LL}^N \left(n,\tilde{\ell},\frac{1}{R-a}\right) $ on both sides, and using Lemma \ref{LemmaLiebLinigerNeumannLowerBound} on the left-hand side, we find\begin{equation}
		\begin{aligned}
			&n\frac{\pi^2}{3}\rho^2\left(1+2\rho a+\text{const. }(\rho R)^{6/5}\right)-n\frac{\pi^2}{3}\tilde{\rho}^2\left(1-4\tilde{\rho} (R-a)-\text{const. }n^{-2/3}\right)\\
			&\geq  \left(\frac{1}{16R^2}-E_{LL}^N \left(n,\tilde{\ell},\frac{1}{R-a}\right)\right)(1-\braket{\psi|\psi}),
		\end{aligned}
	\end{equation}
	with $ \tilde{\rho}=n/\tilde{\ell}=\rho/(1-(\rho-1/\ell)R)$.
	Using the upper bound $ E^N_{LL}\left(n,\tilde{\ell},\frac{1}{R-a}\right)\leq n\frac{\pi^2}{3}\tilde{\rho}^2 $ on the right-hand side, as well as $ 2\rho \geq\tilde{\rho}\geq \rho(1+\rho R)$, we find
	\begin{equation}
		\begin{aligned}
			\text{const. }n\rho^2R^2\left(\rho R+n^{-2/3}\right)&\geq \left(\frac{1}{16}-R^2n\frac{4\pi^2}{3}\rho^2\right)\left(1-\braket{\psi|\psi}\right).
		\end{aligned}
	\end{equation}
	It follows that we have \begin{equation}
		\braket{\psi|\psi}\geq 1-\text{const. }\left(n(\rho R)^3+n^{1/3}(\rho R)^2\right).
	\end{equation}
\end{proof}
For $ n\leq \tau (\rho R)^{-9/5} $ with $ \tau=\frac{3}{16\pi^2}\frac{1}{8} $ and $ \rho R\leq \frac{1}{2} $, we find \begin{equation}
	\braket{\psi|\psi}\geq 1-\textnormal{const. }n(\rho R)^3=1-\textnormal{const. }(\rho R)^{6/5}.
\end{equation}
It is now straightforward to show the following result.
\begin{proposition}
	\label{PropositionLowerBoundSpecN}
	For $ n(\rho R)^2\leq  \frac{3}{16\pi^2}\frac{1}{8} $ and $ \rho R\leq \frac{1}{2} $, we have \begin{equation}
		E^N(n,\ell)\geq n\frac{\pi^2}{3}\rho^2\left(1+2\rho a+\textnormal{const. }\left(\frac{1}{n^{2/3}}+n(\rho R)^3+n^{1/3}(\rho R)^2\right)\right).
	\end{equation}
\end{proposition}
\begin{proof}
	By Lemma \ref{LemmaNormBoundEpsilon} with $ \epsilon=1 $, we reduce to a Lieb-Liniger model with volume $ \tilde{\ell} $, density $ \tilde{\rho} $, and coupling $ c $, and we have $ \tilde{\ell}=\ell-(n-1)R $, $ \tilde{\rho}=\frac{n}{\tilde{\ell}} $ and $ c=\frac{2}{R-a} $. Notice that $\rho(1+\rho R)\leq \tilde{\rho}\leq \rho(1+2\rho R)$. Hence, by Lemmas \ref{LemmaLiebLinigerNeumannLowerBound} and \ref{LemmaImprovedMassBound}, \begin{equation}
		\begin{aligned}
			E^N(n,\ell)&\geq E_{LL}^N(n,\tilde{\ell},c)\braket{\psi|\psi}\\&\geq
			n\frac{\pi^2}{3}\rho^2\left(1+2\rho a-\text{const. }\frac{1}{n^{2/3}}\right)\left(1-\text{const. }\left(n(\rho R)^3+n^{1/3}(\rho R)^2\right)\right).
		\end{aligned}
	\end{equation}
\end{proof}
The previous proposition has a corollary that we will also use. 
\begin{corollary} \label{CorollaryLowerBoundSpecN}
	For $ \frac{\tau}{2} (\rho R)^{-9/5}\leq n\leq \tau (\rho R)^{-9/5} $ with $ \tau=\frac{3}{16\pi^2}\frac{1}{8} $ and $ \rho R\leq \frac{1}{2} $, 
	\begin{equation}
 \label{someeqn02}
		E^N(n,\ell)\geq n\frac{\pi^2}{3}\rho^2\left(1+2\rho a-\textnormal{const. }(\rho R)^{6/5}\right).
	\end{equation}
\end{corollary}
\subsection{Lower bound for arbitrary $N$}
\label{seclowboundarbn}


The lower bound in Corollary \ref{CorollaryLowerBoundSpecN} only applies to particle numbers of order $ (\rho R)^{-9/5} $. In this subsection, we generalize to any number of particles by performing a Legendre transformation in the particle number and going to the grand canonical ensemble. First, we justify that only particle numbers of order lesser than or equal to $ (\rho R)^{-9/5} $ are relevant for a certain choice of the chemical potential $ \mu $.
\begin{lemma}\label{LemmaLocalizationFbound}
Let $C$ denote the unspecified constant in \eqref{someeqn02}, $ \tau=\frac{3}{16\pi^2}\frac{1}{8} $, and let $ n=4 \rho \ell m+n_0 $ with $ n_0\in[0,4\rho \ell) $ for some $ m\in\mathbb{N} $. Let $ \mu=\pi^2\rho^2\left(1+\frac{8}{3}\rho a\right)$. Assume that \hbox{$ C(\rho R)^{6/5}<  1/4 $}, that \hbox{$\frac{\tau}{8} (\rho R)^{-9/5} \leq \rho\ell\leq \frac{\tau}{4} (\rho R)^{-9/5} $}, and that $\rho a \geq-1/16$. Then,  \begin{equation}
		E^{N}(n,\ell)-\mu n \geq E^{N}(n_0,\ell)-\mu n_0.
	\end{equation}
\end{lemma}
\begin{proof}
	By Corollary \ref{CorollaryLowerBoundSpecN}, we have \begin{equation}
		E^{N}(4\rho\ell,\ell)\geq\frac{\pi^2}{3}64\ell\rho^3\left(1+8\rho a-C(\rho R)^{6/5}\right).
	\end{equation}
	Superadditivity caused by the positive potential implies \begin{equation}
		E^N(n,\ell)-\mu n\geq m\left(E^N(4\rho\ell,\ell)-4\mu\rho\ell \right)+E^N(n_0,\ell)-\mu n_0.
	\end{equation}
	The result, therefore, follows from the fact that \begin{equation}
		\frac{64}{3}\pi^2\ell\rho^3\left(1+8\rho a-C(\rho R)^{6/5}\right)\geq 4\pi^2\ell\rho^3\left(1+\frac{8}{3}\rho a\right)=4\mu\rho \ell.
	\end{equation}
\end{proof}
\begin{proof}[Proof of Proposition \ref{PropositionLowerBound}]
	Note $\rho=N/L$ in this proposition, and we introduce the notation $\bar{\rho}=n/l$ later in the proof. \\For the case $ N\leq\tau (\rho R)^{-9/5} $, the result follows from Proposition \ref{PropositionLowerBoundSpecN}. For $ N\geq \tau (\rho R)^{-9/5} $, we will divide the system into $M\in \mathbb{N}_+ $ boxes of length $ \ell=L/M $. Choose the number of boxes $ M $ such that $ \frac{\tau}{8}\left(\rho R\right)^{-9/5}\leq \rho\ell\leq \frac{\tau}{4}\left(\rho R\right)^{-9/5} $ and $ \mu=\pi^2\rho^2\left(1+\frac{8}{3}\rho a\right)$. Furthermore, assume that $ C(\rho R)^{6/5}<  1/4 $ and $\rho a\geq -1/16$. (Note that the cases $ C(\rho R)^{6/5}\geq  1/4 $ or $ \rho a<-\frac{1}{16} $ are trivial, by choosing a sufficiently large constant in the error term). 
 
Let $ F^N(\mu,L)=\inf_{N'}\left(E^N(N',L)-\mu N'\right) $ be the Legendre transform. Since $v$ is repulsive, we have that \begin{equation}
\label{EqLocalizationF}
		E^N(N,L)\geq F^N(\mu,L)+\mu N\geq M F^N(\mu,\ell)+\mu N.
	\end{equation}
	 By Lemma \ref{LemmaLocalizationFbound},  \begin{equation}
		F^N(\mu,\ell):=\inf_{n}\left(E^N(n,\ell)-\mu n\right)=\inf_{n<4\rho\ell}\left(E^N(n,\ell)-\mu n\right).
	\end{equation}
	 Corollary \ref{CorollaryLowerBoundSpecN} implies that for $ n<4\rho\ell $ and $\bar{\rho}=n/\ell$, \begin{equation}
		\begin{aligned}
			E^{N}(n,\ell)&\geq \frac{\pi^2}{3}n\bar{\rho}^2\left(1+2\bar{\rho}a\right)-\textnormal{const. }\ell\rho^3(\rho R)^{6/5}.
		\end{aligned}
	\end{equation}
	Thus, we have \begin{equation}
		F^{N}(\mu,\ell)\geq \inf_{\bar{\rho}<4\rho}(g(\bar{\rho})-\mu\bar{\rho})\ell-\textnormal{const. }\ell\rho^3(\rho R)^{6/5},
	\end{equation}
	with $
	g(\bar{\rho})=
	\frac{\pi^2}{3}\bar{\rho}^3\left(1+2\bar{\rho}a\right)
	$. Note that $ g $ is convex and continuously differentiable,  with invertible derivative for $ \bar{\rho} a\geq -\frac{1}{16}  $. Hence, using \eqref{EqLocalizationF}, \begin{equation}
		\begin{aligned}
			E^{N}(N,L)\geq M(F^{N}(\mu,\ell)+\mu \rho\ell)\geq M\frac{\pi^2}{3}\ell \rho^3\left(1+2\rho a-\textnormal{const. }(\rho R)^{6/5}\right)\\
			=\frac{\pi^2}{3} N\rho^2 \left(1+2\rho a-\textnormal{const. }(\rho R)^{6/5}\right),
		\end{aligned}
	\end{equation}
	where the equality follows from the specific choice of $ \mu=g'(\rho) $.
\end{proof}

\section{Anyons and proof of Theorem \ref{TheoremAnyon}}
\label{SectionOtherSymmetries}


In Theorem \ref{TheoremFermion} and below, we discussed the fact that the fermionic ground state energy can be found from Theorem \ref{TheoremMain} by means of a unitary transformation. It was also mentioned that this concept can be generalized to a version of 1D anyonic symmetry \cite{leinaas1977theory,bonkhoff2021bosonic,posske2017second}. We will now define our interpretation of such anyons, depending on a statistical parameter $\kappa\in[0,\pi]$ that defines the phase $e^{i\kappa}$ accumulated upon particle exchange. We also include a Lieb--Liniger interaction of strength $2c>0$, such as in \cite{kundu1999exact,hao2008ground,batchelor2006one}.

To start, divide the configuration space into sectors $ \Sigma_\sigma:=\{x_{\sigma_1}<x_{\sigma_2}<\dots<x_{\sigma_N}\}\subset \R^N $ indexed by permutations $ \sigma=(\sigma_1,\dots,\sigma_N) $, and the diagonal 
$\Delta_N:=\bigcup_{1\leq i<j\leq N}\{x_i=x_j\}$. Consider the kinetic energy operator on $\R^N\setminus\Delta_N$,
\begin{equation}
	H_N=-\sum_{i=1}^{N}\partial_{x_i}^2,
\end{equation}
with domain \begin{equation}
	\label{eqdom}
	\begin{aligned}
		\mathcal{D}(H_N)=\bigg\{\varphi=\euler{-i\frac{\kappa}{2}\Lambda(x)}f(x)&\ \bigg\vert\ f \text{ is continuous, symmetric in $x_1,\dots,x_N$, smooth on each $\Sigma_\sigma$,}\\&\quad \ \text{and } (\partial_i-\partial_j)\varphi\rvert^{ij}_+-(\partial_i-\partial_j)\varphi\rvert^{ij}_-=2c\ \euler{-i\frac{\kappa}{2}\Lambda(x)} f\rvert^{ij}_0 \text{ for all }i\neq j \bigg\}.
	\end{aligned}
\end{equation}
Here, $ \vert^{ij}_{\pm} $ and $\vert^{ij}_{0}$ mean the function should be evaluated for $ x_i\to x_j^{\pm}$ and $x_i=x_j$ respectively, and  
\begin{equation}
	\Lambda(x):= \sum_{i<j}\epsilon(x_i-x_j)\hspace{1cm}\text{with}\hspace{1cm} \epsilon(x)=\begin{cases}
		1&\text{for }x>0\\
		-1&\text{for }x<0\\
		0&\text{for }x=0
	\end{cases}.
\end{equation} The idea is that the (perhaps rather artificial) boundary condition in \eqref{eqdom} encodes the presence of a delta potential of strength $2c$, just like it would for bosons. 

\begin{proposition}\label{PropositionAnyonQuadraticForm}
	Let $0<\kappa<\pi$. $ H_N $ is symmetric, with corresponding quadratic form \begin{equation}
		\mathcal{E}_{\kappa,c}(\varphi)=\sum_{i=1}^{N}\int_{{\R^N\setminus\Delta_N}} \abs{\partial_{x_i}\varphi(x)}^2+\frac{2c}{\cos(\kappa/2)}\sum_{i<j} \delta(x_i-x_j)\abs{\varphi(x)}^2\diff x.
	\end{equation}
\end{proposition}
\begin{proof}
	Let $ \varphi,\vartheta\in \mathcal{D}(H_N) $. Then, by partial integration, \begin{equation}
		\begin{aligned}
			\braket{\vartheta\vert H_N \varphi}&=-\sum_{i=1}^{N}\int_{\R^N\setminus\Delta_N}\overline{\vartheta} \partial_{x_i}^2\varphi\\&=\sum_{i=1}^{N}\int_{\R^N\setminus\Delta_N}\overline{\partial_{x_i}\vartheta}\partial_{x_i}\varphi-\int_{\R^{N-1}\setminus\Delta_{N-1}}\sum_{i\neq j}\left(\overline{\vartheta}\partial_{x_i}\varphi\vert^{ij}_--\overline{\vartheta}\partial_{x_i}\varphi\vert^{ij}_+\right)\\
			&=\sum_{i=1}^{N}\int_{\R^N\setminus\Delta_N}\overline{\partial_{x_i}\vartheta}\partial_{x_i}\varphi+\int_{\R^{N-1}\setminus\Delta_{N-1}}\sum_{i< j}\left(\overline{\vartheta}(\partial_{x_i}-\partial_{x_j})\varphi\vert^{ij}_+-\overline{\vartheta}(\partial_{x_i}-\partial_{x_j})\varphi\vert^{ij}_-\right).
		\end{aligned}
	\end{equation}
	Let $ f,g$ be the functions such that $ \varphi=\euler{-i\frac{\kappa}{2}\Lambda}f $ and $ \vartheta=\euler{-i\frac{\kappa}{2}\Lambda}g $. Then,
	
	\begin{equation}
		\begin{aligned}
			\braket{\vartheta\vert H_N \varphi}&=\sum_{i=1}^{N}\int_{\R^N\setminus\Delta_N}\overline{\partial_{x_i}\vartheta}\partial_{x_i}\varphi+\int_{\R^{N-1}\setminus\Delta_{N-1}}\sum_{i< j}\left(\overline{g}(\partial_{x_i}-\partial_{x_j})f\vert^{ij}_+-\overline{g}(\partial_{x_i}-\partial_{x_j})f\vert^{ij}_-\right)\\
			&=\sum_{i=1}^{N}\int_{\R^N\setminus\Delta_N}\overline{\partial_{x_i}\vartheta}\partial_{x_i}\varphi+\int_{\R^{N-1}\setminus\Delta_{N-1}}2\sum_{i< j}\left(\overline{g}(\partial_{x_i}-\partial_{x_j})f\vert^{ij}_+\right),
		\end{aligned}
	\end{equation}
	where the last step follows from the symmetry of $f$. Note that the boundary conditions of the functions in $ \mathcal{D}(H_N) $ imply \begin{equation}
		\begin{aligned}
		    (\partial_i-\partial_j)\varphi\rvert^{ij}_+-(\partial_i-\partial_j)\varphi\rvert^{ij}_-&=\euler{-i\frac{\kappa}{2}\left(1+S\right)}(\partial_i-\partial_j)f\rvert^{ij}_+-\euler{-i\frac{\kappa}{2}\left(-1+S\right)}(\partial_i-\partial_j)f\rvert^{ij}_-=2c \varphi\rvert^{ij}_0\\&=\euler{-i\frac{\kappa}{2}S}2c f\rvert^{ij}_0,
		\end{aligned}
	\end{equation}
	where $ S:=\Lambda-\epsilon(x_i-x_j) $, so that $S$ does not depend on $x_i-x_j$. By symmetry of $ f $, it follows that \begin{equation}
		\begin{aligned}
			\euler{-i\frac{\kappa}{2}\left(1+S\right)}&(\partial_i-\partial_j)f\rvert^{ij}_+-\euler{-i\frac{\kappa}{2}\left(-1+S\right)}(\partial_i-\partial_j)f\rvert^{ij}_-\\
			&=\euler{-i\frac{\kappa}{2}\left(1+S\right)}(\partial_i-\partial_j)f\rvert^{ij}_++\euler{-i\frac{\kappa}{2}\left(-1+S\right)}(\partial_i-\partial_j)f\rvert^{ij}_+\\
			&=\euler{-i\frac{\kappa}{2}S}2\cos(\kappa/2)(\partial_i-\partial_j)f\rvert^{ij}_+\\
			&=\euler{-i\frac{\kappa}{2}S}2c f\rvert^{ij}_0,
		\end{aligned}
	\end{equation}
	so that \begin{equation}
		2(\partial_i-\partial_j)f\rvert^{ij}_+=\frac{2c}{\cos(\kappa/2)}f\rvert^{ij}_0. 
	\end{equation}
	Hence, it follows that \begin{equation}\label{EqQuadraticFormDerivation}
		\braket{\vartheta\vert H_N \varphi}=\sum_{i=1}^{N}\int_{{\R^N\setminus\Delta_N}}\overline{\partial_{x_i}\vartheta} \partial_{x_i}\varphi(x)+\frac{2c}{\cos(\kappa/2)}\sum_{i<j} \delta(x_i-x_j)\overline{\vartheta(x)}\varphi(x)\diff x.
	\end{equation}
	Starting from $ \braket{H_N\vartheta\vert \phi} $, we can arrive at \eqref{EqQuadraticFormDerivation} by the same steps, proving that $ H_N $ is symmetric. 	
\end{proof}
\begin{remark}\label{RemarkAnyons}
	Since $ \mathcal{E}_{\kappa,c}$ is non-negative and closable, it follows that $ H_N $ has a self-adjoint Friedrichs extension. This is what we regard as the Hamiltonian of the 1D anyon gas with statistical parameter $ \kappa $ and Lieb--Liniger interaction of strength $2c\delta_0 $ that is relevant for Theorem \ref{TheoremAnyon}.
\end{remark}
We are now ready to prove Theorem \ref{TheoremAnyon} as outlined in Section \ref{SecOthersymmetries}.
\begin{proof}[Proof of Theorem \ref{TheoremAnyon}]
	Let $ \mathcal{E}_c $ denote the bosonic quadratic form with potential $ v_c=\tilde{v}+2c\delta_0 $. By Proposition \ref{PropositionAnyonQuadraticForm} and the observation that the quadratic form is independent of the phase factors, we see that the unitary operator $ U_\kappa: f\mapsto \euler{-i\frac{\kappa}{2}\Lambda}f $ provides a unitary equivalence of the bosonic and anyonic set-ups. That is, $ U_\kappa\dom{\mathcal{E}_{0,c/\cos(\kappa/2)}}=\dom{\mathcal{E}_{\kappa,c}} $ with $ \mathcal{E}_{\kappa,c}(U_\kappa f)=\mathcal{E}_{0,c/\cos(\kappa/2)}(f) $. Hence, the result follows from Theorem \ref{TheoremMain}.
\end{proof}

\backmatter

%
%
%

\bmhead{Acknowledgments}
JA and JPS were partially supported by the Villum Centre of Excellence for the Mathematics of Quantum Theory (QMATH, Grant No. 10059). RR was supported by the European Research Council (ERC) under the European Union's Horizon 2020 research and innovation programme (ERC CoG UniCoSM, Grant Agreement No. 724939). JA is grateful to IST Austria for its hospitality during a visit  and to Robert Seiringer for interesting discussions. RR thanks the University of Copenhagen for the hospitality during a visit.

\section*{Declarations}
%
 \bmhead{Competing interests} There are no competing interests.


%
%
%
%

\begin{appendices}
\section{Some facts about the scattering length}
\label{AppendixA}
Here, we collect some useful facts about the scattering 
length, and prove Lemma \ref{lemscatlength2}.
\begin{enumerate}
\item  \textit{The scattering length for compact potentials defined in Lemma \ref{lemscatlength} is independent of the choice of $R\geq R_0$ (where $R_0$ is the range of the potential)}. This follows from the uniqueness of the scattering solution: let $R_0\leq R_1<R_2$. Then there are scattering solutions $f_1$ on $[-R_1,R_1]$ with scattering length $a_1$ and $f_2$ on $[-R_2,R_2]$ with scattering length $a_2$. By uniqueness of the scattering solution on $[-R_1,R_1]$ and the asymptotic behavior of $f_2$, we have $f_1=\frac{R_2-a_2}{R_1-a_2}f_2\mathbbm{1}_{[-R_1,R_1]}$, so that $a_1=a_2$.

\item
    \textit{If $v\neq 0$ then $a>-\infty$}. This can be seen by the fact that the only function with vanishing kinetic energy on the interval $[-R,R]$ is the constant function, but this function has potential energy $\int_{[-R,R]} v>0$. Hence $\frac{1}{R-a}>0$ and $a>-\infty$. 
\item \textit{The scattering solution is positive.} This follows from the fact that $\abs{f}$ has energy smaller than or equal to that of $f$, with equality if and only if $\abs{f}=f$.
\item 
    \textit{The scattering solution $f_0$ is radially increasing.} Notice that $f_0\in H^1([-R,R])$, so that it is continuous. Now define on $[0,R]$ the radially increasing function $$f(x)=\min_{y\in[x,R]}(f_0(y)),$$
    and extend it by reflection to $[-R,0)$. If for some $x\in[-R,R]$ it holds that $f(x)\neq f_0(x)$, then \hbox{$f'(x)=0$} as $f$ must necessarily be constant in a neighbourhood of $x$. Thus,  $$
    \int\abs{f'}^2=\int_{\{f=f_0\}}\abs{f'}^2= \int_{\{f=f_0\}}\abs{f_0'}^2\leq \int\abs{f_0'}^2.
    $$
    Furthermore, $f\leq f_0$, so we also have $$
    \int v\abs{f}^2\leq \int v \abs{f_0}^2.
    $$
    By uniqueness of the scattering solution it must follows that $f=f_0$, which can be true only if $f_0$ is radially increasing.
\item
    \textit{If $v$ is supported in $[-R_0,R_0]$, then $a\leq R_0$}. This is clear from the fact that the scattering energy $\frac{4}{R-a}$ is positive for any $R>R_0$.
\item
	In Appendix C of \cite{lieb2006mathematics}, it is unclear whether Lemma \ref{lemscatlength} is proved for any (positive) measure $v$ or not. However, the proof easily allows for any measure by noting that in one dimension, the map $f\mapsto\int v \abs{f}^2$ is $H^1$-weakly lower semicontinuous by Fatou's lemma.
\end{enumerate}

\begin{proof}[\textbf{Proof of Lemma \ref{lemscatlength2}}]
	Let $\tilde{R}>R>R'$ and let $f_0^y$ denote the scattering solution with potential $\mathbbm{1}_{[-y,y]}v$. By Lemma \ref{lemscatlength}, we have \begin{equation}
		\begin{aligned}
			\frac{4}{\tilde{R}-a_R}&=\int^{\tilde{R}}_{-\tilde{R}}2|\partial_xf_0^R|^2+\mathbbm{1}_{[-R,R]}v(x)|f_0^R(x)|^2\diff x\\&\geq\int^{\tilde{R}}_{-\tilde{R}}2|\partial_xf_0^R|^2+\mathbbm{1}_{[-R',R']}v(x)|f_0^R(x)|^2\diff x\geq \frac{4}{\tilde{R}-a_{R'}} .
		\end{aligned}
	\end{equation}
	Thus $a_R$ is increasing in $R$. Now, also assume $R'>-a_{R'}$ (note this will hold if $R'$ is large enough by the fact that $a_R$ increases with $R$). Then, again by Lemma \ref{lemscatlength}, \begin{equation}
		\begin{aligned}
			\frac{4}{\tilde{R}-a_{R'}}+2\int_{R'}^{R} v(x) \abs{f_0^{R'}}^2&=\int^{\tilde{R}}_{-\tilde{R}}2|\partial_xf_0^{R'}|^2+\mathbbm{1}_{[-R,R]}v(x)|f_0^{R'}(x)|^2\diff x\geq\frac{4}{\tilde{R}-a_{R}} .
		\end{aligned}
	\end{equation}
	Using that $f_0^{R'}(x)=(x-a_{R'})/(\tilde{R}-a_{R'})$ for $x>R'$, we find\begin{equation}
		\frac{4(a_R-a_{R'})}{(\tilde{R}-a_{R'})(\tilde{R}-a_{R})}\leq \frac{2}{(\tilde{R}-a_{R'})^2}\int_{R'}^{R}v(x)(x-a_{R'})^2\leq \frac{8}{(\tilde{R}-a_{R'})^2}\int_{R'}^{R}v(x)x^2\diff x.
	\end{equation}
	Multiplying both sides by $\tilde{R}^2$ and taking $\tilde{R},R\to\infty$, we conclude 
	\begin{equation}
		a-a_{R'}\leq 2\int_{R'}^{\infty}v(x)x^2\diff x,
	\end{equation}
	as desired.
\end{proof}

\end{appendices}


\bibliography{bibliography}

\end{document}